\documentclass[10pt,twocolumn,twoside]{IEEEtran} 

\usepackage{amsfonts}
\usepackage{amssymb}
\usepackage{amsmath}
\usepackage{amsthm}
\usepackage{graphicx}
\usepackage{subfigure}
\usepackage{pst-plot}
\usepackage{pst-node}
\usepackage{pst-fill}
\usepackage{color}
\usepackage[verbose,nospace,sort]{cite}
\graphicspath{{figure/}}

\begin{document}
\newtheorem{theorem}{Theorem}
\newtheorem{corollary}{Corollary}
\newtheorem{condition}{Condition}
\newtheorem{lemma}{Lemma}
\theoremstyle{definition}
\newtheorem{example}{Example}
\theoremstyle{definition}
\newtheorem{definition}{Definition}
\theoremstyle{definition}
\newtheorem{remark}{Remark}
\newcommand{\xE}{\textrm{E}}
\newcommand{\xH}{\mathbf{H}}
\newcommand{\xI}{\mathbf{I}}
\newcommand{\xu}{\mathbf{u}}
\newcommand{\xU}{\mathbf{U}}
\newcommand{\xv}{\mathbf{v}}
\newcommand{\xV}{\mathbf{V}}
\newcommand{\xY}{\mathbf{Y}}
\newcommand{\xX}{\mathbf{X}}
\newcommand{\xZ}{\mathbf{Z}}
\newcommand{\xdeg}{\textrm{deg}}

\title{On Feasibility of Interference Alignment in MIMO Interference Networks}

\author{Cenk M. Yetis$^\ast$, \IEEEmembership{Member, IEEE,
} Tiangao Gou, \IEEEmembership{Student Member, IEEE, }\\ Syed A. Jafar
\IEEEmembership{Senior Member, IEEE} and Ahmet H. Kayran,
\IEEEmembership{Senior Member, IEEE}
\thanks{C. M. Yetis is with Informatics Institute, Satellite Communications and Remote Sensing, Istanbul Technical University,
      Maslak, Istanbul, 34469, TURKEY. Email: cenkmyetis@yahoo.com. The author is supported in
      part by The Scientific and Technological Research Council of
      Turkey (TUBITAK). The author is on leave at University of California Irvine.}
\thanks{T. Gou and S. A. Jafar are with Electrical Engineering and Computer Science, University of California
Irvine, Irvine, CA, 92697, USA. Email: \{tgou,syed\}@uci.edu}
\thanks{A. H. Kayran is with Department of Electronics and Communications Engineering, Istanbul Technical University,
      Maslak, Istanbul, 34469, TURKEY.
      Email: kayran@itu.edu.tr.}}

 \maketitle
\begin{abstract}
We explore the feasibility of interference alignment in signal
vector space -- based only on beamforming -- for $K\text{-user}$
MIMO interference channels. Our main contribution is to relate the
feasibility issue to the problem of determining the solvability of a
multivariate polynomial system, considered extensively in algebraic
geometry. It is well known, e.g. from Bezout's theorem, that generic
polynomial systems are solvable if and only if the number of
equations does not exceed the number of variables. Following this
intuition, we classify signal space interference alignment problems
as either proper or improper based on the number of equations and
variables. Rigorous connections between feasible and proper systems
are made through Bernshtein's theorem for the case where each
transmitter uses only one beamforming vector. The \mbox{multi-beam}
case introduces dependencies among the coefficients of a polynomial
system so that the system is no longer generic in the sense required
by both theorems. In this case, we show that the connection between
feasible and proper systems can be further strengthened (since the
equivalency between feasible and proper systems does not always
hold) by including standard information theoretic outer bounds in
the feasibility analysis.
\end{abstract}


\section{Introduction}

The degrees of freedom (DoF) of wireless interference networks
represent the number of interference-free signaling-dimensions in
the network. In a network with $K$ transmitters and $K$ receivers
and non-degenerate channel conditions, it is well known that $K$
non-interfering spatial signaling dimensions can be created if the
transmitters or the receivers are able to jointly process their
signals. Until recently it was believed that with distributed
processing at transmitters and receivers, it is not possible to
resolve these signaling dimensions so that only one degree of
freedom is available. However, the discovery of a new idea called
interference alignment has shown that the DoF of wireless
interference networks can be much higher \cite{85}.

\subsection{Evolution of Interference Alignment}
Interference alignment refers to the consolidation of multiple
interfering signals into a small subspace at each receiver so that
the number of interference-free dimensions remaining for the desired
signal can be maximized. The idea evolved out of the DoF studies for
the 2-user X channel \cite{88,113} and has since been applied to a
variety of networks in increasingly sophisticated forms. The
majority of interference alignment schemes proposed so far, fall
into one of two categories -- (1) signal space alignment and (2)
signal level alignment.
\subsubsection{Interference Alignment in Signal Vector Space}
\par  The potential for
overlapping interference spaces was first pointed out by Maddah-Ali
et. al. in \cite{123,111} where iterative schemes were formulated
for optimizing transmitters and receivers in conjunction with dirty
paper coding/successive decoding schemes. The idea of interference
alignment was crystallized in a report by Jafar  \cite{125} where
the first  explicit (closed form, non-iterative) and linear (no
successive-decoding or dirty paper coding) interference alignment
scheme in signal vector space was presented. The explicit linear
approach introduced by Jafar in \cite{125} was adopted by Maddah-Ali
et. al. in their subsequent report and journal paper \cite{112,
113}, while \cite{125} developed into the journal paper by Jafar and
Shamai \cite{88}. Interference alignment was also independently
discovered by Weingarten et. al. \cite{104} in the context of the
compound multiple input single output (MISO) broadcast channel (BC).

Following the early success on the X channel and the compound MISO
BC, signal space interference alignment schemes were introduced for
the $K\text{-user}$ interference channel with equal (unequal) number
of antennas at all transmitters and receivers by Cadambe and Jafar
(Gou and Jafar) in \cite{85} (\hspace{-0.025cm}\cite{89}), for X
networks with arbitrary number of users by Cadambe and Jafar in
\cite{99}, for cellular networks by Suh and Tse in \cite{105}, for
MIMO bidirectional relay networks (Y channel) by Lee and Lim in
\cite{107}, for ergodic fading interference networks by Nazer et.
al. in \cite{106}, and for interference networks with secrecy
constraints in \cite{126}. Interference networks with constant
channel coefficients posed a barrier for signal space interference
alignment schemes because they did not provide distinct rotations of
vector spaces on each link that were needed for linear interference
alignment. The problem was circumvented to a certain extent for
complex interference channels in \cite{100}, where phase rotations
were exploited in a similar manner through the use of asymmetric
complex signaling. However, for constant and \emph{real} channel
coefficients, these linear alignment schemes were not sufficient and
a different class of alignment schemes based on structured (e.g.
lattice) codes that align interference in signal scale were
introduced.

\subsubsection{Interference Alignment in Signal Scale}
The first interference alignment scheme in signal scale was
introduced for the many-to-one interference channel by Bresler et.
al. in \cite{108} and for fully connected interference networks  by Cadambe
et. al. in \cite{109}. Unlike random codes for which decoding the
\emph{sum} of interfering signals is equivalent to decoding each of
the interfering signals, these schemes rely on codewords with a
lattice structure, which opens the possibility that the sum of
interfering signals can be decoded even when the individual
interfering signals are not decodable. This is because the sum of
lattice points is another lattice point, and hence may be decoded as
a valid codeword. Lattice based alignment schemes were further
investigated for interference networks by Sridharan et. al. in
\cite{102,103} and for networks with secrecy constraints by He and
Yener in \cite{117}. An interesting interference alignment scheme in
signal scale was introduced by Etkin and Ordentlich in \cite{101}.
This work used fundamental results from diophantine approximation
theory to prove that the rational and irrational scaled versions of
a lattice ``stood apart'' from each other, and thus could be
separated. The result was extended to almost all irrational numbers
by Maddah-Ali et. al. in \cite{113} by translating the notion of
linear independence (exploited in linear interference alignment
schemes) into the notion of \emph{rational independence} in signal
scale. With this new insight, the asymptotic alignment scheme of Cadambe
and Jafar from \cite{85} was essentially adopted in \cite{113} to
achieve interference alignment in signal scale and following the
approach in \cite{85}, was shown to approach the DoF outer bound.

In spite of the obvious advantages of signal scale alignment schemes
(especially those based on rational independence \cite{101,113}) for
obtaining DoF characterizations,  a downside to these schemes is
that they seem to bring to light primarily the artifacts of the
infinite SNR regime and offer  little in terms  of useful insights
for the practical setting with finite SNR and finite precision
channel knowledge, where the notion of rational independence loses
its relevance. Signal space alignment schemes on the other hand, are
desirable both for their analytical tractability as well as the
useful insights they offer for finite SNR regime  where they may be
naturally combined with selfish approaches \cite{124}. Within the
class of signal vector space interference alignment schemes,
alignment in spatial dimension through multiple antennas (MIMO) is
found to be more robust to practical limitations such as frequency
offsets than alignment in time or frequency dimensions \cite{110}.
However, the feasibility of linear interference alignment for
general MIMO interference networks remains an open problem
\cite{55}. It is this problem - the feasibility of linear
interference alignment for MIMO interference networks - that we
address in this paper. We explain our objective through the
following examples.

\subsection{The Feasibility Question - Examples}
\subsubsection{Symmetric Systems}
\par Let $(M\times N, d)^K$ denote the $K\text{-user}$ MIMO
interference network, where every transmitter has $M$ antennas,
every receiver has $N$ antennas and each user wishes to achieve $d$
DoF. We call such a system a symmetric system. Consider the
following examples.

\begin{itemize}
\item $(2\times 2, 1)^3$ - It is shown in \cite{85} that for the 3-user interference
network with 2 antennas at each node, each user can achieve 1 DoF by
presenting a closed form solution for linear interference alignment,
i.e., by linear beamforming at the transmitters and linear combining
at the receivers. However, is there a way to analytically determine
the feasibility of this system without finding a closed form
solution?
\item $(5\times 5,2)^4$ - Consider the 4-user interference network with 5 antennas at each
user and we wish to achieve 2 DoF per user for a total of 8 DoF. A
theoretical solution to this problem is not known but numerical
evidence in \cite{55} clearly indicates that a linear interference
alignment solution exists. Numerical algorithms are one way to
determine the feasibility of linear interference alignment. However,
is there a way to analytically determine the feasibility of
alignment without running the numerical simulation?
\end{itemize}

\subsubsection{Asymmetric Systems}
Let us introduce the notation $\left(M^{[1]}\times N^{[1]}, d^{[1]}\right)\cdots(M^{[K]}\times
N^{[K]}, d^{[K]})$ to denote the $K\text{-user}$ MIMO interference
network, where the $k^{th}$ transmitter and receiver have $M^{[k]}$
and $N^{[k]}$ antennas, respectively and the $k^{th}$ user demands
$d^{[k]}$ DoF. We call such a system an asymmetric system. Consider
the following examples.

\begin{itemize}
\item Consider the simple system $(2\times 1, 1)^2$,
which is clearly feasible (proper) because simple zero-forcing is
enough for achievability. However, now consider the ${(2\times 1,
1)(1\times 2,1)}$ system, where the same total number of DoF is
desired. Although these systems have the same number of total
antennas, is the latter system still achievable?
\item Consider the 2-user interference network ${(2\times 3, 1)(3\times 2,
1)}$, where a total of 2 DoF is desired. The achievable scheme for
this system is presented in \cite{93}. Now, consider the same scheme
with increased number of users; that is, the 4-user interference
network $(2\times 3,1)^2(3\times 2,1)^2$, where a total of 4 DoF is
desired. Is this system still achievable, where DoF is doubled by
simply going from two users to four users?
\end{itemize}

In this paper, we address all these questions. Our approach is
to consider the signal space interference alignment problem as the
solvability of a multivariate polynomial system, and then place it into perspective with classical results in algebraic geometry where these problems are extensively studied.

\section{Preliminaries}\label{sec:Basics}

\subsection{System Model}

We consider the same $K\text{-user}$ MIMO interference network as
considered in \cite{55}. The received signal at the $n^{th}$ channel
use can be written as follows:

\begin{eqnarray*}
\xY^{[k]}(n)=\sum_{l=1}^K\xH^{[kl]}(n)\xX^{[l]}(n)+\xZ^{[k]}(n),
~~
\end{eqnarray*}
$\forall k\in\mathcal{K}\triangleq\{1,2,...,K\}$. Here,
$\xY^{[k]}(n) \textrm{ and } \xZ^{[k]}(n)$ are the $N^{[k]}\times 1$
received signal vector and the zero mean unit variance circularly
symmetric additive white Gaussian noise vector (AWGN) at the
$k^{th}$ receiver, respectively. $\xX^{[l]}(n)$ is the
$M^{[l]}\times 1$ signal vector transmitted from the $l^{th}$
transmitter and $\xH^{[kl]}(n)$ is the $N^{[k]}\times M^{[l]}$
matrix of channel coefficients between the $l^{th}$ transmitter and
the $k^{th}$ receiver. $\xE[||\xX^{[l]}(n)||^2]=P^{[l]}$ is the
transmit power of the $l^{th}$ transmitter. Hereafter, we omit the
channel use index $n$ for the sake of simplicity. The DoF for the
$k^{th}$ user's message is denoted by
$d^{[k]}\leq\min(M^{[k]},N^{[k]})$.

As defined earlier, $\left(M\times N, d\right)^K$ denotes the
$K\text{-user}$ symmetric MIMO interference network, where each
transmitter and receiver has $M$ and $N$ antennas, respectively and
each user demands $d$ DoF; therefore, the total DoF demand is $Kd$.
In general, let $\Pi_{k=1}^K\left(M^{[k]}\times N^{[k]},
d^{[k]}\right)=\left(M^{[1]}\times N^{[1]},
d^{[1]}\right)\cdots(M^{[K]}\times N^{[K]}, d^{[K]})$ denote the
$K\text{-user}$ asymmetric MIMO interference network, where the
$k^{th}$ transmitter and receiver have $M^{[k]}$ and $N^{[k]}$
antennas, respectively and the $k^{th}$ user demands $d^{[k]}$ DoF.
Some sample symmetric and asymmetric systems are shown in Fig.
\ref{fig:Figures}. \psset{unit=0.8cm}
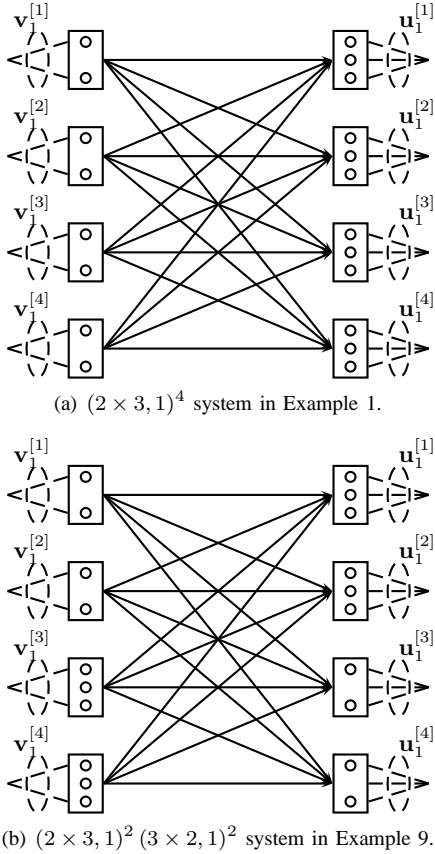
\begin{figure}[!t]
\centering \subfigure[$\left(2\times3,1\right)^4$ system in Example
\ref{ex:2x3,1}.] 
{
    \label{fig:Fig1a}
\begin{pspicture}(0,0)(7,6.4) 
\psframe(1,4.8)(1.6,5.8) \pscircle(1.3,5){.1} \pscircle(1.3,5.6){.1}
\pnode(0,5.3){Tx1}
\pnode(1,5){Tx1a}\pnode(1,5.6){Tx1b}\ncline[linestyle=dashed]{Tx1}{Tx1a}\ncline[linestyle=dashed]{Tx1}{Tx1b}
\pnode(1.6,5.3){Tx1c}\psellipse[linestyle=dashed](.5,5.3)(.2,.5)
\rput[l](.1,6){\small{$\xv^{[1]}_1$}}

\psframe(5.4,4.8)(6,5.8) \pscircle(5.7,5){.1} \pscircle(5.7,5.3){.1}
\pscircle(5.7,5.6){.1} \pnode(7,5.3){Rx1}
\pnode(6,5){Rx1a}\pnode(6,5.3){Rx1b}\pnode(6,5.6){Rx1c}\ncline[linestyle=dashed]{Rx1}{Rx1a}\ncline[linestyle=dashed]{Rx1}{Rx1b}\ncline[linestyle=dashed]{Rx1}{Rx1c}
\pnode(5.4,5.3){Rx1d} \psellipse[linestyle=dashed](6.5,5.3)(.2,.5)
\rput[l](6.5,6){\small{$\xu^{[1]}_1$}}

\psframe(1,3.2)(1.6,4.2) \pscircle(1.3,3.4){.1} \pscircle(1.3,4){.1}
\pnode(0,3.7){Tx2}
\pnode(1,3.4){Tx2a}\pnode(1,4){Tx2b}\ncline[linestyle=dashed]{Tx2}{Tx2a}\ncline[linestyle=dashed]{Tx2}{Tx2b}
\pnode(1.6,3.7){Tx2c}\psellipse[linestyle=dashed](.5,3.7)(.2,.5)
\rput[l](.1,4.4){\small{$\xv^{[2]}_1$}}

\psframe(5.4,3.2)(6,4.2) \pscircle(5.7,3.4){.1}
\pscircle(5.7,3.7){.1} \pscircle(5.7,4){.1} \pnode(7,3.7){Rx2}
\pnode(6,3.4){Rx2a}\pnode(6,3.7){Rx2b}\pnode(6,4){Rx2c}\ncline[linestyle=dashed]{Rx2}{Rx2a}\ncline[linestyle=dashed]{Rx2}{Rx2b}\ncline[linestyle=dashed]{Rx2}{Rx2c}
\pnode(5.4,3.7){Rx2d} \psellipse[linestyle=dashed](6.5,3.7)(.2,.5)
\rput[l](6.5,4.4){\small{$\xu^{[2]}_1$}}

\psframe(1,1.6)(1.6,2.6) \pscircle(1.3,1.8){.1}
\pscircle(1.3,2.4){.1} \pnode(0,2.1){Tx3}
\pnode(1,1.8){Tx3a}\pnode(1,2.4){Tx3b}\ncline[linestyle=dashed]{Tx3}{Tx3a}\ncline[linestyle=dashed]{Tx3}{Tx3b}
\pnode(1.6,2.1){Tx3c}\psellipse[linestyle=dashed](.5,2.1)(.2,.5)
\rput[l](.1,2.8){\small{$\xv^{[3]}_1$}}

\psframe(5.4,1.6)(6,2.6) \pscircle(5.7,1.8){.1}
\pscircle(5.7,2.1){.1} \pscircle(5.7,2.4){.1} \pnode(7,2.1){Rx3}
\pnode(6,1.8){Rx3a}\pnode(6,2.1){Rx3b}\pnode(6,2.4){Rx3c}\ncline[linestyle=dashed]{Rx3}{Rx3a}\ncline[linestyle=dashed]{Rx3}{Rx3b}\ncline[linestyle=dashed]{Rx3}{Rx3c}
\pnode(5.4,2.1){Rx3d} \psellipse[linestyle=dashed](6.5,2.1)(.2,.5)
\rput[l](6.5,2.8){\small{$\xu^{[3]}_1$}}

\psframe(1,0)(1.6,1) \pscircle(1.3,.2){.1} \pscircle(1.3,.8){.1}
\pnode(0,.5){Tx4}
\pnode(1,.2){Tx4a}\pnode(1,.8){Tx4b}\ncline[linestyle=dashed]{Tx4}{Tx4a}\ncline[linestyle=dashed]{Tx4}{Tx4b}
\pnode(1.6,.5){Tx4c}\psellipse[linestyle=dashed](.5,.5)(.2,.5)
\rput[l](.1,1.2){\small{$\xv^{[4]}_1$}}

\psframe(5.4,0)(6,1) \pscircle(5.7,.2){.1} \pscircle(5.7,.5){.1}
\pscircle(5.7,.8){.1} \pnode(7,.5){Rx4}
\pnode(6,.2){Rx4a}\pnode(6,.5){Rx4b}\pnode(6,.8){Rx4c}\ncline[linestyle=dashed]{Rx4}{Rx4a}\ncline[linestyle=dashed]{Rx4}{Rx4b}\ncline[linestyle=dashed]{Rx4}{Rx4c}
\pnode(5.4,.5){Rx4d} \psellipse[linestyle=dashed](6.5,.5)(.2,.5)
\rput[l](6.5,1.2){\small{$\xu^{[4]}_1$}}

\ncline{->}{Tx1c}{Rx1d}\ncline{->}{Tx1c}{Rx2d}\ncline{->}{Tx1c}{Rx3d}\ncline{->}{Tx1c}{Rx4d}
\ncline{->}{Tx2c}{Rx2d}\ncline{->}{Tx2c}{Rx1d}\ncline{->}{Tx2c}{Rx3d}\ncline{->}{Tx2c}{Rx4d}
\ncline{->}{Tx3c}{Rx3d}\ncline{->}{Tx3c}{Rx1d}\ncline{->}{Tx3c}{Rx2d}\ncline{->}{Tx3c}{Rx4d}
\ncline{->}{Tx4c}{Rx4d}\ncline{->}{Tx4c}{Rx1d}\ncline{->}{Tx4c}{Rx2d}\ncline{->}{Tx4c}{Rx3d}
\end{pspicture}
} \centering
\subfigure[$\left(2\times3,1\right)^2\left(3\times2,1\right)^2$ system in Example \ref{ex:2x3,3x2}.] 
{
    \label{fig:Fig1b}
\begin{pspicture}(0,0)(7,6.4) 

\psframe(1,4.8)(1.6,5.8) \pscircle(1.3,5){.1} \pscircle(1.3,5.6){.1}
\pnode(0,5.3){Tx1}
\pnode(1,5){Tx1a}\pnode(1,5.6){Tx1b}\ncline[linestyle=dashed]{Tx1}{Tx1a}\ncline[linestyle=dashed]{Tx1}{Tx1b}
\pnode(1.6,5.3){Tx1c}\psellipse[linestyle=dashed](.5,5.3)(.2,.5)
\rput[l](.1,6){\small{$\xv^{[1]}_1$}}

\psframe(5.4,4.8)(6,5.8) \pscircle(5.7,5){.1} \pscircle(5.7,5.3){.1}
\pscircle(5.7,5.6){.1} \pnode(7,5.3){Rx1}
\pnode(6,5){Rx1a}\pnode(6,5.3){Rx1b}\pnode(6,5.6){Rx1c}\ncline[linestyle=dashed]{Rx1}{Rx1a}\ncline[linestyle=dashed]{Rx1}{Rx1b}\ncline[linestyle=dashed]{Rx1}{Rx1c}
\pnode(5.4,5.3){Rx1d} \psellipse[linestyle=dashed](6.5,5.3)(.2,.5)
\rput[l](6.5,6){\small{$\xu^{[1]}_1$}}

\psframe(1,3.2)(1.6,4.2) \pscircle(1.3,3.4){.1} \pscircle(1.3,4){.1}
\pnode(0,3.7){Tx2}
\pnode(1,3.4){Tx2a}\pnode(1,4){Tx2b}\ncline[linestyle=dashed]{Tx2}{Tx2a}\ncline[linestyle=dashed]{Tx2}{Tx2b}
\pnode(1.6,3.7){Tx2c}\psellipse[linestyle=dashed](.5,3.7)(.2,.5)
\rput[l](.1,4.4){\small{$\xv^{[2]}_1$}}

\psframe(5.4,3.2)(6,4.2) \pscircle(5.7,3.4){.1}
\pscircle(5.7,3.7){.1} \pscircle(5.7,4){.1} \pnode(7,3.7){Rx2}
\pnode(6,3.4){Rx2a}\pnode(6,3.7){Rx2b}\pnode(6,4){Rx2c}\ncline[linestyle=dashed]{Rx2}{Rx2a}\ncline[linestyle=dashed]{Rx2}{Rx2b}\ncline[linestyle=dashed]{Rx2}{Rx2c}
\pnode(5.4,3.7){Rx2d} \psellipse[linestyle=dashed](6.5,3.7)(.2,.5)
\rput[l](6.5,4.4){\small{$\xu^{[2]}_1$}}

\psframe(1,1.6)(1.6,2.6) \pscircle(1.3,1.8){.1}
\pscircle(1.3,2.1){.1} \pscircle(1.3,2.4){.1} \pnode(0,2.1){Tx3}
\pnode(1,1.8){Tx3a}\pnode(1,2.4){Tx3b}\ncline[linestyle=dashed]{Tx3}{Tx3a}\ncline[linestyle=dashed]{Tx3}{Tx3b}
\pnode(1.6,2.1){Tx3c}\psellipse[linestyle=dashed](.5,2.1)(.2,.5)
\rput[l](.1,2.8){\small{$\xv^{[3]}_1$}}

\psframe(5.4,1.6)(6,2.6) \pscircle(5.7,1.8){.1}
\pscircle(5.7,2.4){.1} \pnode(7,2.1){Rx3}
\pnode(6,1.8){Rx3a}\pnode(6,2.1){Rx3b}\pnode(6,2.4){Rx3c}\ncline[linestyle=dashed]{Rx3}{Rx3a}\ncline[linestyle=dashed]{Rx3}{Rx3b}\ncline[linestyle=dashed]{Rx3}{Rx3c}
\pnode(5.4,2.1){Rx3d} \psellipse[linestyle=dashed](6.5,2.1)(.2,.5)
\rput[l](6.5,2.8){\small{$\xu^{[3]}_1$}}

\psframe(1,0)(1.6,1) \pscircle(1.3,.2){.1}\pscircle(1.3,.5){.1}
\pscircle(1.3,.8){.1} \pnode(0,.5){Tx4}
\pnode(1,.2){Tx4a}\pnode(1,.8){Tx4b}\ncline[linestyle=dashed]{Tx4}{Tx4a}\ncline[linestyle=dashed]{Tx4}{Tx4b}
\pnode(1.6,.5){Tx4c}\psellipse[linestyle=dashed](.5,.5)(.2,.5)
\rput[l](.1,1.2){\small{$\xv^{[4]}_1$}}

\psframe(5.4,0)(6,1) \pscircle(5.7,.2){.1} \pscircle(5.7,.8){.1}
\pnode(7,.5){Rx4}
\pnode(6,.2){Rx4a}\pnode(6,.5){Rx4b}\pnode(6,.8){Rx4c}\ncline[linestyle=dashed]{Rx4}{Rx4a}\ncline[linestyle=dashed]{Rx4}{Rx4b}\ncline[linestyle=dashed]{Rx4}{Rx4c}
\pnode(5.4,.5){Rx4d} \psellipse[linestyle=dashed](6.5,.5)(.2,.5)
\rput[l](6.5,1.2){\small{$\xu^{[4]}_1$}}

\ncline{->}{Tx1c}{Rx1d}\ncline{->}{Tx1c}{Rx2d}\ncline{->}{Tx1c}{Rx3d}\ncline{->}{Tx1c}{Rx4d}
\ncline{->}{Tx2c}{Rx2d}\ncline{->}{Tx2c}{Rx1d}\ncline{->}{Tx2c}{Rx3d}\ncline{->}{Tx2c}{Rx4d}
\ncline{->}{Tx3c}{Rx3d}\ncline{->}{Tx3c}{Rx1d}\ncline{->}{Tx3c}{Rx2d}\ncline{->}{Tx3c}{Rx4d}
\ncline{->}{Tx4c}{Rx4d}\ncline{->}{Tx4c}{Rx1d}\ncline{->}{Tx4c}{Rx2d}\ncline{->}{Tx4c}{Rx3d}
\end{pspicture}
}
 \caption{Sample symmetric and asymmetric systems.}
 \label{fig:Figures}
\end{figure}

\subsection{Interference Alignment in Signal Space - Beamforming and Zero Forcing Formulation}\label{subsec:recalign}

In interference alignment precoding, the transmitted signal from the
$k^{th}$ user is $\xX^{[k]}=\xV^{[k]}\tilde{\xX}^{[k]}$, where
$\tilde{\xX}^{[k]}$ is a $d^{[k]}\times 1$ vector that denotes the
$d^{[k]}$ independently encoded streams transmitted from the
$k^{th}$ user. The $M^{[k]}\times d^{[k]}$ precoding (beamforming) filters
$\xV^{[k]}$ are designed to maximize the overlap of interference
signal subspaces at each receiver while ensuring that the desired
signal vectors at each receiver are linearly independent of the
interference subspace. Therefore, each receiver can zero-force all
the interference signals without zero-forcing any of the desired
signals. The zero-forcing filters at the receiver are denoted by
$\xU^{[k]}$. In \cite{55}, it is shown that an interference
alignment solution requires the simultaneous satisfiability of the
following conditions:

\begin{eqnarray}
\hspace{-.4cm}\xU^{[k]\dagger}\xH^{[kj]}\xV^{[j]} &=& 0, \forall j\neq k \textrm{ ~~~and }\label{eqn:condition1}\\
\hspace{-.4cm}\mbox{rank}\left(\xU^{[k]\dagger}\xH^{[kk]}\xV^{[k]}\right)&=&d^{[k]},
~~\forall k\in\{1,2,...,K\}, \label{eqn:condition2}
\end{eqnarray}
where $^\dagger$ denotes the conjugate transpose operator. Very
importantly, \cite{55} explains how the condition
\eqref{eqn:condition2} is automatically satisfied almost surely if
the channel matrices do not have any special structure,
${\mbox{rank}(\xU^{[k]})=\mbox{rank} (\xV^{[k]})=d^{[k]}\leq
\min(M^{[k]},N^{[k]}})$ and $\xU^{[k]}, \xV^{[k]}$ are designed to
satisfy \eqref{eqn:condition1}, which is independent of all direct
channels ${\xH^{[kk]}}$. We assume that general MIMO channels have
no structure and we force the transmit and receive filters to
achieve the required ranks by design. Thus,
\eqref{eqn:condition2} is automatically satisfied for us as well.

\section{Proper System} \label{sec:Propernetworks}
Based on classical results in algebraic geometry, like Bezout's
theorem, it is well known  that a generic system of multivariate
polynomial equations is solvable if and only if the number of
equations does not exceed the number of variables. While the
qualification  ``generic system of polynomials" is intended in a
precise sense and limits the scope of settings where the result can
be rigorously applied, the intuition behind this statement is
believed to be much more widely true. This conventional wisdom forms
the starting point for our work. By accurately accounting for the
number of equations, $N_e$, and the number of variables, $N_v$, we
classify a signal space interference alignment problem as either
improper or proper, depending on whether or not the number of
equations exceeds the number of variables.

\subsection{Counting the Total Number of Equations and Variables}
We rewrite
the conditions in \eqref{eqn:condition1} as follows:
\begin{equation}
\xu^{[k]\dagger}_m\xH^{[kj]}\xv^{[j]}_n = 0,~~ j\neq k, ~
j, k\in\{1,2,...,K\} \label{eqn:condition1b}
\end{equation}
\begin{equation*}
\forall n\in\{1,2,...,d^{[j]}\} \textrm{ and } \forall
m\in\{1,2,...,d^{[k]}\}
\end{equation*}
where $\xv^{[j]}_n$ and $\xu^{[k]}_m$ are the transmit and receive
beamforming vectors (columns of precoding and interference
suppression filters, respectively).
\par $N_e$ is directly obtained from \eqref{eqn:condition1b} as follows:
\begin{equation*}
N_e=\underset{{\substack{ k,j\in\mathcal{K} \\
k\neq j}}}{{\sum }}d^{[k]}d^{[j]}.
\end{equation*}
\par However, calculating the number of variables $N_v$ is less
straightforward. In particular, we have to be careful not to count
any superfluous variables that do not help with interference
alignment.
\par At the $k^{th}$
transmitter, the number of $M^{[k]}\times 1$ transmit beamforming
vectors to be designed is $d^{[k]}$ $\left(\xv^{[k]}_n\textrm{, }
\forall n\in\{1,2,...,d^{[k]}\}\right)$. Therefore, at first sight,
it may seem that the precoding filter of the $k^{th}$ transmitter,
$\xV^{[k]}$, has $d^{[k]}M^{[k]}$ variables. However, as we argue
next, we can eliminate $(d^{[k]})^2$ of these variables without loss
of generality.
\par The $d^{[k]}$ linearly independent columns of transmit precoding matrix ${\bf V}^{[k]}$  span the transmitted
signal space
\begin{eqnarray*}
\mathcal{T}^{[k]}&=&\mbox{span}({\bf V}^{[k]})\\
&=&\{{\bf v}: \exists {\bf a}\in\mathbb{C}^{d^{[k]}\times 1}, ~{\bf v}={\bf V}^{[k]}{\bf a} \}.
\end{eqnarray*}
Thus, the columns of ${\bf V}^{[k]}$ are the basis for the
transmitted signal space. However, the basis representation is not
unique for a given subspace. In particular, consider any full rank
$d^{[k]}\times d^{[k]}$ matrix ${\bf B}$. Then, continuing from the
last step of the above equations,
\begin{eqnarray*}
\mathcal{T}^{[k]}&=&\{{\bf v}: \exists {\bf a}\in\mathbb{C}^{d^{[k]}\times 1}, ~{\bf v}={\bf V}^{[k]}{\bf B}^{-1}{\bf Ba} \}\\
&=&\mbox{span}({\bf V}^{[k]}{\bf B}^{-1}).
\end{eqnarray*}
Thus, post-multiplication of the transmit precoding matrix with any
invertible matrix on the right does not change the transmitted
signal subspace. Suppose that we choose ${\bf B}$ to be the
$d^{[k]}\times d^{[k]}$ matrix that is obtained by deleting the
bottom $M^{[k]}-d^{[k]}$ rows of ${\bf V}^{[k]}$. Then, we have
${\bf V}^{[k]}{\bf B}^{-1}=\tilde{\xV}^{[k]}$, which is a
$M^{[k]}\times d^{[k]}$ matrix with the following structure:
\begin{equation*}
\tilde{\xV}^{[k]}=\left[
\begin{array}{ccccc}
 &  & \xI_{d^{[k]}}  &  &\\
\bar{\xv}_1^{[k]} & \bar{\xv}_2^{[k]} & \bar{\xv}_3^{[k]}  & \cdots  & \bar{\xv}_{d^{[k]}}^{[k]}%
\end{array}%
\right]
\end{equation*}
where $\xI_{d^{[k]}}$ is the $d^{[k]}\times d^{[k]}$ identity matrix
and ${\bar{\xv}_n^{[k]}, \forall n\in\{1,2,...,d^{[k]}\}}$ are
$\left(M^{[k]}-d^{[k]}\right)\times 1$ vectors. It is easy to argue
that there is no other basis representation for the transmitted
signal space with fewer variables.

Therefore, by eliminating superfluous variables for the interference alignment problem, the
number of variables to be designed for the precoding filter of
the $k^{th}$ transmitter, $\tilde{\xV}^{[k]}$, is
$d^{[k]}\left(M^{[k]}-d^{[k]}\right)$. Likewise, the actual number
of variables to be designed for the interference suppression
filter of the $k^{th}$ receiver, $\tilde{\xU}^{[k]}$, is
$d^{[k]}\left(N^{[k]}-d^{[k]}\right)$. As a result, the total number
of variables in the network to be designed is:
\begin{equation*}
N_v=\sum_{k=1}^Kd^{[k]}\left(M^{[k]}+N^{[k]}-2d^{[k]}\right).
\end{equation*}

\subsection{Proper System Characterization}

To formalize the definition of a proper system, we first introduce
some notation. We use the notation $E_{mn}^{kj}$ to represent the
equation
\begin{eqnarray*}
\xu^{[k]\dagger}_m\xH^{[kj]}\xv^{[j]}_n = 0.
\end{eqnarray*}
The set of variables involved in an equation $E$ is indicated by the
function $\mbox{var}(E)$. Clearly
\begin{eqnarray*}
|\mbox{var}(E_{mn}^{kj})|=(M^{[j]}-d^{[j]})+(N^{[k]}-d^{[k]}),
\end{eqnarray*}
where $|\cdot|$ is the cardinality of a set.

Using this notation, we denote the set of $N_e$ equations as follows:
\begin{eqnarray*}
\mathcal{E}&=&\{E_{mn}^{kj}|~ j,k\in\mathcal{K}, k\neq j, \\
&&~~~~~~~~~m\in\{1,\cdots, d^{[k]}\}, n\in\{1,\cdots,d^{[j]}\}\}.
\end{eqnarray*}
This leads us to the formal definition of a proper system.
\begin{definition} \label{def:proper}
\theoremstyle{definition} A $\Pi_{k=1}^K(M^{[k]}\times
N^{[k]},d^{[k]})$ system is proper if and only if
\begin{equation}\forall S\subset\mathcal{E}, |S|\leq\left| \bigcup_{E\in
S}\mbox{var}(E)\right|. \label{eqn:proper}\end{equation}
\end{definition}

In other words, for all subsets of equations, the number of
variables involved must be at least as large as the number of
equations in that subset.

The condition to identify a proper system can be computationally cumbersome
because we have to test all subsets of equations. However, several simplifications are possible in this regard. We start with symmetric systems.

\subsection{Symmetric Systems $(M\times N,d)^K$}
For symmetric systems, simply comparing the total number of
equations and the total number of variables suffices to determine
whether the system is proper or improper.
\begin{theorem}\label{theo:symmetric}
A symmetric system $(M\times N,d)^K$ is proper if and only if
\begin{equation*}
N_v\geq N_e \Rightarrow M+N-(K+1)d\geq 0.
\end{equation*}
\end{theorem}
\begin{proof} Because of the symmetry, each equation involves the
same number of variables and any deficiency in the number of
variables shows up in the comparison of the total number of
variables versus the total number of equations. Plugging in the
values of $N_v$ and $N_e$ computed earlier, we have the result of
Theorem \ref{theo:symmetric}. \end{proof}

\example \label{ex:2x3,1} Consider the $\left(2\times 3, 1\right)^4$
system. For this system, $M+N-(K+1)d=2+3-(5)=0$ so that this system
is proper.

\example Consider the $(1\times 2, 1)^3$ system, i.e., a 3-user
symmetric interference network, where each transmitter has one
antenna, each receiver has two antennas, and each user demands 1
DoF.  For this system, $M+N-(K+1)d = 1+2-(4)<0$ so that this system
is improper.

\begin{remark} In light of the intuition that proper systems are likely to be feasible,
Theorem \ref{theo:symmetric} implies that for every user to achieve
$d$ DoF in a $K\text{-user}$ symmetric network, it suffices to have
a total of $M+N\geq (K+1)d$ antennas between the transmitter and
receiver of a user. The antennas can be distributed among the
transmitter and receiver arbitrarily as long as each of them has at
least $d$ antennas and as long as the symmetric nature of the system
is preserved. In particular, to achieve $K$ DoF in a $K\text{-user}$
symmetric network (1 DoF per user), we only need a total of $K+1$
antennas between the transmitter and receiver of a user.
\end{remark}

\example Consider a $4$-user symmetric network, where we wish to
achieve 4 DoF. Then, 5 antennas between the transmitter and receiver
of a user would suffice to produce a proper system, e.g., the system $(2\times 3,1)^4$ in
Example \ref{ex:2x3,1}.

\example Consider a $6$-user symmetric network, where we wish to
achieve 6 DoF. Then, 7 antennas between the transmitter and receiver
of a user would suffice, e.g., $(3\times 4,1)^6$.
\par The following corollary shows the limitations of linear
interference alignment over constant MIMO channels (with no symbol
extensions).

\begin{corollary}
The DoF of a proper $(M\times N, d)^K$ system, which is normalized
by a single user's DoF in the absence of interference, is upper
bounded by:
\begin{equation*}
\frac{dK}{\min(M,N)}\leq
1+\frac{\max(M,N)}{\min(M,N)}-\frac{d}{\min(M,N)}.
\end{equation*}
\end{corollary}
\begin{proof} The proof is straightforward from the condition of
Theorem \ref{theo:symmetric}.\end{proof}

\begin{remark} \label{rem:Rem2} For the
case $M=N$, note that the DoF of a proper system is no more than
twice the DoF achieved by each user in the absence of interference.
Note that for diagonal (time-varying) channels, it is shown in
\cite{85} that the DoF of a $K\text{-user}$ MIMO network ($M=N$
antennas at each node) is $K/2$ times the number of DoF achieved by
each user in the absence of interference. This result shows that the
diagonal structure of the channel matrix is very helpful. Going from
the case of no structure (general MIMO channels) to diagonal
structure, the ratio of total DoF to the single user DoF increases
from a maximum value of 2 to $K/2$.
\end{remark}

The following corollary identifies the groups of symmetric systems,
which are either all proper or all improper.
\begin{corollary}
If $(M\times N, d)^K$ system is proper (improper) then so is the
$\big((M+1)\times(N-1),d\big)^K$ system as long as
$d\leq\min(M,N-1)$. Similarly, if  the $(M\times N, d)^K$ system is
proper (improper) then so is the $\big((M-1)\times(N+1),d\big)^K$
system as long as $d\leq\min(M-1,N)$.
\end{corollary}
\begin{proof} Since the condition in Theorem \ref{theo:symmetric} depends
only on $M+N$, it is clear that we can transfer transmit and receive
antennas without affecting the proper (or improper) status of the
system. \end{proof}

\example The systems $(1\times 4, 1)^4, (2\times 3,1)^4,(3\times
2,1)^4,$ and ${(4\times 1,1)^4}$ are in the same group, which are
formed by successively transferring an antenna between transmitters
and receivers. It is easy to see that the $(4\times 1,1)^4$ system
is proper because simple zero-forcing suffices to achieve the DoF
demand. By virtue of being in the same group, the rest are proper as
well.

\example By similar arguments, the systems $(1\times 3,1)^3,$
${(2\times 2, 1)^3, \textrm{ and } (3\times 1,1)^3}$ are in the same
group and are all proper.

\subsection{Asymmetric Systems $\Pi_{k=1}^K\left(M^{[k]}\times N^{[k]}, d^{[k]}\right)$ }

For asymmetric systems, if the system is improper, simply comparing
the total number of equations and the total number of variables may
suffice. \begin{theorem} \label{theo:asymmetric} An asymmetric
system $\Pi_{k=1}^K(M^{[k]}\times N^{[k]},d^{[k]})$ is improper if
\begin{equation}
N_v<N_e \Leftrightarrow \sum_{k=1}^Kd^{[k]}\left(M^{[k]}+N^{[k]}-2d^{[k]}\right)<\sum_{{\substack{ k,j\in\mathcal{K} \\
k\neq j}}}^Kd^{[k]}d^{[j]}. \label{eqn:generalcondition}
\end{equation}
\end{theorem}

\example Consider the system ${(2\times 2,1)(2\times 3,1)^3}$, which
is clearly infeasible when we compare it to the ${(2\times 3,1)^4}$
system in Example \ref{ex:2x3,1}. Confirmatively, the former system
is improper since it has 11 variables and 12 equations in total.

Note that we can sometimes identify the bottleneck equations in the
system by checking the equations with the fewest number of
variables, i.e., the equations involving the fewest number of
transmitter and receiver antennas.

 \example
\label{ex:2x1,1} Consider the simple system $(2\times 1, 1)^2$,
which is clearly feasible (proper) because simple zero-forcing is
enough for achievability. However, now consider the $(2\times 1,
1)(1\times 2,1)$ system, which also has the same total number of
equations $N_e$ and variables $N_v$ as the $(2\times 1,1)^2$ system.
Thus, only comparing $N_v$ and $N_e$ would mislead one to believe
that this system is proper. However, suppose that we only check the
equation $E_{11}^{12}$; that is, our subset is $S=\{E_{11}^{12}\}$
so that $|S|=1$. $E_{11}^{12}$ corresponds to the link between the
transmitter 2 and receiver 1, both of which have only one antenna
each. Therefore, $|\mbox{var}(E_{11}^{12})|=0$. Thus, this system
has an equation with zero variable, which makes the system improper;
therefore, infeasible.

\example \label{ex:2x3,3x2}Several interesting cases emerge from
applying the condition \eqref{eqn:generalcondition}. For example,
consider the 2-user interference network $(2\times 3, 1)(3\times 2,
1)$, where a total of 2 DoF is desired. It is easily checked that
this system is proper and the achievable scheme is described in
\cite{93}. Now, consider the \mbox{4-user} interference network,
which consists of two sets of these networks, all interfering with
each other $(2\times 3,1)^2(3\times 2,1)^2$, where a total of 4 DoF
is desired. By using \eqref{eqn:generalcondition}, it is easily
verified that this is a proper system. Surprisingly, by simply going
from two users to four users, DoF is doubled in this case. We also
present the closed form solution for interference alignment of this
system in Section \ref{sec:NewClosedForm}.

\section{Numerical Results}\label{sec:NumResults}

We tested numerous interference alignment problems for both
symmetric and asymmetric cases by using the numerical algorithm in
\cite{55}. In every case so far, we have found the numerical results
to be consistent with the guiding intuition of this work; that is,
for single beam cases, proper systems are almost surely feasible and
improper systems are not.

In this section, we provide numerical results for a few interesting
and representative cases. The results are in terms of the
interference percentage, which is defined in \cite{55}. i.e., the
fraction of the interference power that is existent in the
dimensions reserved for the desired signal. The interference
percentage at the $k^{th}$ receiver is evaluated as follows:
\begin{equation}\label{eqn:leakage}
p_{k}=\frac{\underset{j=1}{\overset{d^{[k]}}{\sum }}\lambda
_{j}\left[\mathbf{Q}^{[k]}\right] }{\textrm{Tr}[\mathbf{Q}^{[k]}]},
\end{equation}
where $\lambda_j$ denotes the smallest eigenvalue of a matrix, Tr
denotes the trace of a matrix, and $\mathbf{Q}^{[k]}$ denotes the
interference covariance matrix at the $k^{th}$ receiver:
\begin{equation*}
\mathbf{Q}^{[k]}=\sum_{j=1,j\neq
k}^K\frac{P^{[j]}}{d^{[j]}}\xH^{[kj]}\xV^{[j]}\xV^{[j]\dagger}\xH^{[kj]\dagger}.
\end{equation*}

The numerator and the denominator of \eqref{eqn:leakage} are the
interference and desired signal space powers at the $k^{th}$
receiver, respectively.

In Fig. \ref{fig:DoFs}, the interference percentages versus the
total number of beams are shown. The total number of beams starts
from the expected total DoF of each network. Therefore, after the
first point on the x-axis, where excess total DoF is demanded the
interference percentage of each network is not zero. The nonzero
interference percentage indicates that interference alignment is not
possible for the demanded total DoF.

Therefore, by observing zero interference percentages on the DoF
point in Fig. \ref{fig:DoFs}, we show that the numerical results are
consistent with our statements in Section \ref{sec:Propernetworks}
that these networks are proper, and thus feasible. Note that in Fig.
\ref{fig:DoFs}, there are numerical results also for multi-beam
cases, which we discuss in Section \ref{sec:GOuterB}.

From the excess total DoF results in Fig. \ref{fig:DoFs}, we also
understand that the first two systems with expected 4 total DoF have
more interference percentages than other systems with expected 8
total DoF. We also observe that the system with the less total
number of antennas at the receiver side has more interference
percentage than other system with the same expected total DoF, e.g.,
${(2\times 3,1)^2(3\times 2,1)^2}$ has more interference percentage
than $(2\times 3,1)^4$.
\begin{figure} [htb]
\centering \hspace{-1.5cm}\includegraphics[height=8cm, width=10cm]
{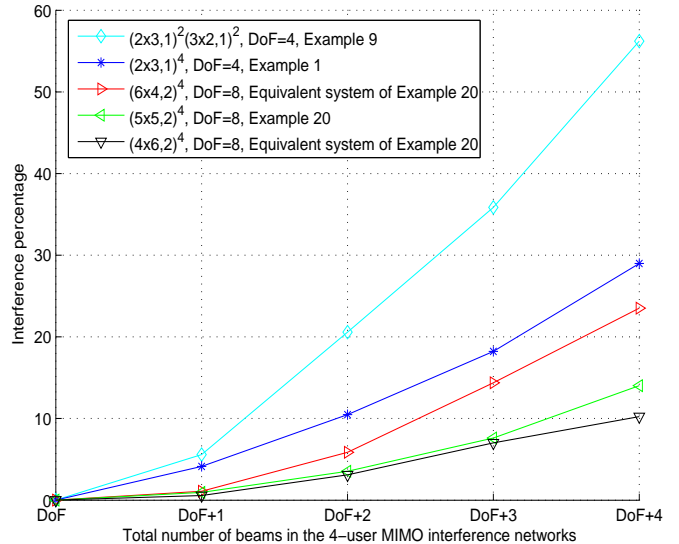} \hspace{-2cm} \caption{Interference percentages as a function
of the total number of beams in the networks (DoF: Expected total
degrees of freedom. DoF+i (i=1,2,3,4): Excess total degrees of
freedom). \label{fig:DoFs}}
\end{figure}

\section{Bezout's and Bernshtein's Theorems}\label{sec:BTheorem}

As explained earlier, the definition of proper systems is an
intuitive generalization of the classical result in algebraic
geometry known as Bezout's theorem. While the formal statement of
Bezout's theorem is presented later, the theorem essentially states
that ``generic'' systems of multivariate polynomial equations are
solvable if and only if the number of equations does not exceed the
number of variables. Since the notion of generic system is critical
in our work, we first summarize what is meant by a generic system in
simple terms. While the mathematical definition of ``genericity'' is
presented in Appendix, the notion of a ``generic'' system refers to
two aspects.
\begin{enumerate}
\item The supports of polynomials, which are determined by non-zero coefficients of polynomials.
\item The independence (e.g. algebraic independence) of non-zero coefficients.
\end{enumerate}
``Generic'' system in the sense of Bezout's theorem refers to the
system of dense polynomials (all coefficients up to the degree of
each polynomial are non-zero) with independent random coefficients.
According to Bezout's theorem, these systems are solvable almost
surely as long as the number of polynomial equations does not exceed
the number of variables, and the number of solutions is equal to the
product of the degrees of polynomials.

On the other hand, ``generic'' system in the sense of Bernshtein's
theorem refers to the system of sparse polynomials with independent
random coefficients. For dense polynomials, the result of Bezout's
theorem can be derived from Bernshtein's theorem and therefore,
Bernshtein generalizes Bezout's theorem. For our feasibility
problem, the system of polynomials does not always satisfy
``genericity'' in Bernshtein's theorem since while the coefficients
are independent in the single-beam case, the same is not true for
the multi-beam case.

The single beam case refers to the scenario when each user demands
only one DoF, to be achieved by sending one beam from each
transmitter. Therefore, for this scenario, the channel matrix
$\xH^{[kj]}$ of each user occurs only once in the corresponding
polynomial system \eqref{eqn:condition1}. On the other hand, for a
multi-beam case, consider the user who wishes to achieve more than 1
DoF. The channel matrix $\xH^{[kj]}$ of that user occurs more than
once in the corresponding polynomial system, which leads to
dependent coefficients.

As mentioned before, although we use only Bernshtein's theorem for
the proofs in our work, we also summarize Bezout's theorem as an
elementary step. We briefly rephrase these two theorems insofar as
required within the scope of this paper. Let us start with
definitions and notations.

\subsection{Multivariate Polynomial Systems}

\subsubsection{A polynomial system and its support sets}
\par Let $\mathbb{C}\left[x_1,\cdots,x_n\right]$ denote the
polynomial ring, where the coefficients are in the field
$\mathbb{C}$ and the variable $x_i, \forall i\in\{1,2,\cdots,n\}$
has nonnegative integer (denoted by $\mathbb{Z}_{\geq0}$) exponent.
The multivariate polynomial system that we are interested in
consists of $n$ variables and $n$ equations:
\begin{equation}\label{eqn:polysys}
    f_1=0,\cdots,f_n=0\textrm{,}
\end{equation}
where $f_1,\cdots,f_n \in \mathbb{C}\left[x_1,\cdots,x_n\right]$.
\par Let $e_{il}^j$ denote the nonnegative integer exponent of
the $l^{th}$ variable $x_l$ in the $j^{th}$ monomial of the
polynomial ${f_i,\forall i\in\{1,2,...,n\}}$:
\begin{equation*}
f_{i}=\cdots +c_{ij}\underset{j^{th}\text{ monomial}}{\underbrace{%
x_{1}^{e_{i1}^{j}}x_{2}^{e_{i2}^{j}}\cdots x_{l}^{e_{il}^{j}}\cdots
x_{n}^{e_{in}^{j}}}}+\cdots,
\end{equation*}
where $c_{ij}$ denotes the complex valued coefficient.

Also, let
\begin{equation*}
a_{ij}\triangleq(e_{i1}^j,e_{i2}^j,\cdots,e_{in}^j)\in\mathbb{Z}_{\geq0}^n,
\end{equation*}
\begin{equation*}
\forall i\in\{1,2,...,n\}\textrm{ and }\forall j\in\{1,2,...,m_i\}
\end{equation*}
denote a nonnegative integer vector, which is also called an
exponent vector. \par Then, we denote the $j^{th}$ monomial in $f_i$
as follows:
$$x^{a_{ij}}\triangleq x_1^{e_{i1}^j}x_2^{e_{i2}^j}\cdots
x_n^{e_{in}^j}.$$
\par Finally, let $\mathcal{A}_i=\{a_{i1},\cdots,a_{im_i}\}\subset\mathbb{Z}_{\geq0}^n$
denote the set of exponent vectors with nonzero coefficients in
$f_i$. $\mathcal{A}_i$ is also called the support set of $f_i$.
\par Therefore, each polynomial has the following
structure with a support set $\mathcal{A}_i$:
\begin{equation}\label{eqn:polystruct}
    f_i=\overset{m_{i}}{\underset{j=1}{\sum
    }}c_{ij}x^{a_{ij}}.
\end{equation}
\example Consider the following $i^{th}$ polynomial:
\begin{equation*}
f_i=c_{i1}x_1+c_{i2}x_1x_2+c_{i3}.
\end{equation*}
Then, $a_{i1}=(1,0),a_{i2}=(1,1)\textrm{, and } a_{i3}=(0,0)$.
Accordingly, the support set $\mathcal{A}_i$ for this polynomial is
the set of vertexes of a right triangle.

\subsubsection{Common solutions of a polynomial system}

Let ${S_k=\{x_1^k,\cdots,x_n^k\}}$ denote the $k^{th}$, $\forall
k\in\{1,2,...,s\}$ common solution for the $n$ dimensional
polynomial system \eqref{eqn:polysys}, which has $s$ common
solutions in total:
\begin{equation*}
f_1(x_1^k,\cdots,x_n^k)=0,\cdots,f_n(x_1^k,\cdots,x_n^k)=0.
\end{equation*}
Then, the set of all common solutions $S_C$ that satisfies the
polynomial system \eqref{eqn:polysys} is as follows:
\begin{equation*}
S_C=\{S_1,\cdots,S_s\}.
\end{equation*}
In other words, there are $s$ points in the corresponding space that
satisfy the polynomial system \eqref{eqn:polysys}, e.g.,
${(x_1^k,\cdots,x_n^k)\in \mathbb{C}^n, \forall k\in\{1,2,...,s\}}$.
\subsubsection{The degree of a polynomial}
\par Let $\xdeg(f_i)$ denote the degree of $f_i$, which is defined as follows:
\begin{equation*}
\xdeg(f_i)=\textrm{max}\left({e_{i1}^1+\cdots+e_{in}^1,\cdots,e_{i1}^{m_i}+\cdots
+e_{in}^{m_i}}\right).
\end{equation*}

\subsection{Dense and Sparse Polynomial Systems}

For a dense polynomial system, in any polynomial $f_i$, monomials
with all combinations of variable exponents up to $\xdeg(f_i)$ have
\emph{nonzero} coefficients. On the other hand, for a sparse
polynomial system, in any polynomial $f_i$, some certain monomials
\emph{may} have \emph{zero} coefficients.

\example \label{ex:densepoly} For $n=2$, $\xdeg(f_1)=3$ and
$\xdeg(f_2)=4$, a dense polynomial system is as follows:
\begin{eqnarray*}
  f_1 &=&
  c_{11}x_1^3+c_{12}x_2^3+c_{13}x_1^2x_2+c_{14}x_1x_2^2+c_{15}x_1^2+\\
  &&c_{16}x_2^2+c_{17}x_1x_2+c_{18}x_1+c_{19}x_2+c_{110}\\
  f_2 &=& c_{21}x_1^4+c_{22}x_2^4+c_{23}x_1^3x_2+c_{24}x_1x_2^3+c_{25}x_1^2x_2^2+ \\
  &&c_{26}x_1^3+c_{27}x_2^3+c_{28}x_1^2x_2+c_{29}x_1x_2^2+\\
&&c_{210}x_1^2+c_{211}x_2^2+c_{212}x_1x_2+\\
&&c_{213}x_1+c_{214}x_2+c_{215}.
\end{eqnarray*}

\example \label{ex:sparsepoly} One of the sparse polynomial systems
corresponding to the previous example may be as follows:
\begin{eqnarray*}
  f_1 &=&
  c_{11}x_1x_2^2+c_{12}x_1^2+c_{13}x_2^2+c_{13}\\
  f_2 &=& c_{21}x_1^3x_2+c_{22}x_2^4+c_{23}x_1x_2.
\end{eqnarray*}

Now, we are ready to state Bezout's theorem.

\subsection{Bezout's Theorem}

\begin{theorem}[Bezout's Theorem - specialized]
Given dense polynomials $f_1,\cdots , f_n \in
\mathbb{C}\left[x_1,\cdots,x_n\right]$ with common solutions in
$\mathbb{C}^n$, let $\xdeg(f_i)$ be the degree of polynomial $f_i$.
For independent random coefficients\footnote{Also called generic
choices of coefficients or almost all specializations of
coefficients in mathematics terminology, which we explain in the
Appendix.} $c_{ij},$ ${\forall i\in\{1,2,...,n\} \text{ and }
\forall j\in\{1,2,...,m_i\}}$, the number of common solutions is
exactly equal to $\xdeg(f_1)\xdeg(f_2)\cdots\xdeg(f_n)$.
\end{theorem}

According to Bezout's theorem, the number of common solutions is
$\xdeg(f_1)\xdeg(f_2)=12$ for the Example \ref{ex:densepoly}; that
is, $s=12$.

When the polynomial system is sparse, Bezout's theorem gives a loose
upper bound, which is still 12 for the Example \ref{ex:sparsepoly}.
On the other hand, Bernshtein's theorem gives a tighter result 9
(this result is exact when the coefficients are independent random
variables) as we will show next.

\subsection{Bernshtein's Theorem}\label{subsec:Berns}
Chapter 7 of \cite{94} (hereafter, we briefly refer as \cite{94}) is
recommended for an excellent introduction and for further
information for Bernshtein's theorem. Here, we first briefly
summarize the rudiments of this theorem.

\subsubsection{Newton Polytopes}

Let $\mathbb{C}^*$ denote the complex field excluding zeros,
$\mathbb{C}^*=\mathbb{C} \backslash \{0\}$. A polytope is the convex
hull of a finite set in $\mathbb{R}^n$ and a polytope with integer
coordinates is called lattice polytope. A Newton polytope is a
lattice polytope defined for a polynomial, which is based on the
exponent vectors of monomials with nonzero coefficients:
\begin{equation*}
    P_i=\textrm{Conv}\left(\mathcal{A}_i\right),
\end{equation*}
where Conv(.) denotes the convex hull of a finite set. \example
\label{ex:SupportSets}The support sets of $f_1$ and $f_2$ in Example
\ref{ex:sparsepoly} are
\begin{eqnarray*}
\mathcal{A}_1&=&\{(1,2),(2,0),(0,2),(0,0)\}\textrm{ and }\\
\mathcal{A}_2&=&\{(3,1),(0,4),(1,1)\},
\end{eqnarray*} respectively. The corresponding Newton polytopes (also called the supports of polynomials) are shown in
Fig. \ref{fig:MinkowskiSum}. \psset{unit=0.8cm}
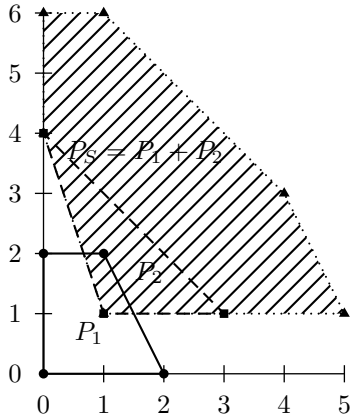
\begin{figure} [!t] \centering
{
\hspace{1cm} 
 \begin{pspicture}(0,-1)(5,6) 
\psaxes[linewidth=.5pt](5,6) \pspolygon[showpoints=true]
(0,0)(0,2)(1,2)(2,0) \pspolygon[linestyle=dashed,
showpoints=true,dotstyle=square*] (1,1)(0,4)(3,1)
\pspolygon[fillcolor=lightgray,fillstyle=hlines,linestyle=dotted,dotsep=2pt,showpoints=true,dotstyle=triangle*]
(0,6)(1,6)(4,3)(5,1)(1,1)(0,4) \rput[br](1,.5){$P_1$}
\rput[br](2,1.5){$P_2$} \rput[br](3,3.5){$P_S=P_1+P_2$}
\end{pspicture}}
 \caption{Minkowski sum of two Newton polytopes. \label{fig:MinkowskiSum}}
\end{figure}
\subsubsection{Mixed Volume and Minkowski Sum}
\par The mixed volume of Newton polytopes includes Minkowski sum operation of Newton
polytopes, which can be carried on by summing two Newton polytopes
at a time. For example, the Minkowski sum of three Newton polytopes
${P_S=P_1+P_2+P_3}$ can be evaluated in two steps, e.g.,
${P_S=P_{S12}+P_3}$, where ${P_{S12}=P_1+P_2}$. The Minkowski sum of
two Newton polytopes is based on the Minkowski sum of their support
sets, i.e., ${\mathcal{A}_S=\mathcal{A}_1+\mathcal{A}_2}$. Minkowski
sum of two sets is basically adding every element of
${\mathcal{A}_1=\{a_{11},\cdots,a_{1m_1}\}}$ to every element of
${\mathcal{A}_2=\{a_{21},\cdots,a_{2m_2}\}}$:
\begin{equation*}
    \mathcal{A}_S=\{a_{1j}+a_{2k}:a_{1j}\in\mathcal{A}_1\textrm{ and
    }a_{2k}\in\mathcal{A}_2\}.
\end{equation*}
\example The Minkowski sum of two sets in Example
\ref{ex:SupportSets} is as follows:
\begin{eqnarray*}
    \mathcal{A}_S&=&\{(4,3),(1,6),(2,3),\\
    &&(5,1),(2,4),(3,1),\\
    &&(3,3),(0,6),(1,3),\\
    &&(3,1),(0,4),(1,1)\}.
\end{eqnarray*}
Therefore, the Minkowski sum of corresponding two Newton polytopes
$P_S=P_1+P_2$ is found as follows:
\begin{equation*}
    P_S=\textrm{Conv}\left(\mathcal{A}_S\right),
\end{equation*}
which is also shown in Fig. \ref{fig:MinkowskiSum}.

The mixed volume of Newton polytopes has the following general
formula: \par $\textrm{MV}(P_{1},\cdots,P_{n})=$
\begin{equation*}
\underset{k=1}{\overset{n}{\sum }}\left(
-1\right)^{n-k}\underset{\begin{array}{c}
I\subset \{1,\cdots,n\} \\
|I|=k\end{array}}{\sum }\textrm{Vol}\left( \underset{i\in
I}{\sum}P_{i}\right),
\end{equation*}
where Vol(.) and MV(.) denote the volume and mixed volume operators,
respectively. $\sum_{i\in I}P_i$ denotes the Minkowski sum of Newton
polytopes. It can be shown that mixed volume always has a
nonnegative value \cite{94}.
\par As a simple example, consider the mixed
volume of two Newton polytopes:
\begin{equation*}
    \textrm{MV}(P_1,P_2)=-\textrm{Vol}(P_1)-\textrm{Vol}(P_2)+\textrm{Vol}(P_S),
\end{equation*}
where $P_S=P_1+P_2$.
\par Therefore, the mixed volume for the system in Example
\ref{ex:sparsepoly} is found as follows:
\begin{equation*}
    \textrm{MV}(P_1,P_2)=-3-3+15=9.
\end{equation*}
\begin{theorem}[Bernshtein's Theorem - specialized]
Given polynomials $f_1,\cdots , f_n \in
\mathbb{C}\left[x_1,\cdots,x_n\right]$ with common solutions in
$(\mathbb{C}^*)^n$, let $P_i$ be the Newton polytope of $f_i$ in
$\mathbb{R}^n$. For independent random coefficients $c_{ij}$,
$\forall i\in\{1,2,...,n\}$ and $\forall j\in\{1,2,...,m_i\}$, the
number of common solutions is exactly equal to the mixed volume of
Newton polytopes, $\textrm{MV}(P_1,\cdots,P_n)$.
\end{theorem}
\par It can be shown that Bezout's theorem is a special case
of Bernshtein's theorem. That is, mixed volume for an $n$
dimensional dense polynomial system is equal to
$\xdeg(f_1)\xdeg(f_2)\cdots\xdeg(f_n)$ \cite{94}.
\subsection{High Dimensional Polynomial Systems}\label{subsec:High}
\par For high dimensional polynomial systems, there is a nice connection between the facets (an $n-1$
dimensional face is called facet for an $n$ dimensional polytope)
and the volumes of polytopes, which significantly simplifies the
computation of mixed volume. For example, by using facets, the mixed
volume for the following simple 3 dimensional polynomial system is
easily found to be 0:
\begin{eqnarray*}
    f_1=c_{11}x_1^2+c_{12}x_2^2+c_{13}x_3+c_{14}&=&0\\
    f_2=c_{21}x_1^2+c_{22}&=&0\\
    f_3=c_{31}x_1+c_{32}&=&0,
\end{eqnarray*}
where clearly, there is no solution when the coefficients are random
variables. We leave the details of facet approach to \cite{94} since
this is a further detail beyond our scope.
\par Computing the mixed volume by using facets is still cumbersome when the system is a little more complicated even for 3 dimensional polynomial systems.
Therefore, there are several \emph{theoretical approaches} in the
mathematics literature that \emph{lead to algorithms} to compute the
mixed volumes, e.g., \cite{96}. These softwares provide rigorous
mixed volume results for polynomial systems. In the next section, we
use the softwares mentioned in \cite{94} to compute the mixed
volumes for some important cases.

 \section{Rigorous Connection Between
Proper and Feasible Systems - Bezout's and Bernshtein's Theorems}
\label{sec:Direct}

As mentioned in the previous sections, we can use Bernshtein's
theorem in order to indirectly show that the corresponding
polynomial system for a single beam case is solvable (not solvable)
almost surely if the mixed volume for that system is nonzero (zero).
Once again, note that the coefficients must be generic in order to
use Bernshtein's theorem. Next, we apply this theorem for some
important systems.

\example For the systems $(2\times 3, 1)^4$ and ${(2\times 3,
1)^2(3\times 2, 1)^2}$, the mixed volumes are 9 and 8, respectively.
In other words, these polynomial systems with independent random
coefficients are solvable almost surely since for each system, mixed
volume is nonzero.

\example \label{ex:infeasible} Now, consider the system ${(2\times
2,1)^3(3\times 5,1)}$, which is infeasible according to the
simulation result. Since the subset of equations, which is obtained
by shutting down the fourth receiver has 9 equations and 8
variables, this system is improper. The mixed volume for this system
is 0. In other words, the corresponding polynomial system with
independent random coefficients is not solvable almost surely.
\par Note once again that we only provide mixed volumes for only some
important cases and also note that mixed volume computation is
${\#\textrm{P-complete}}$ \cite{94}.
\par Bernshtein's theorem applies to a polynomial system with independent random coefficients and with
equal number of equations and variables ($N_e=N_v$). When the number
of equations is greater than the number of variables ($N_e>N_v$), it
can be argued that there is no solution almost surely by using
Bernshtein's theorem as follows. First, note that the equations in
the polynomial system are independent for a single beam case (the
coefficients are independent random variables) and with general MIMO
channels (the polynomial system has no structure). Second, suppose
that we apply Bernshtein's theorem to only $N_v$ polynomials of all
$N_e$ polynomials. Since the mixed volume for these $N_v$
polynomials is finite, the number of solutions for these $N_v$
polynomials is also finite. If the mixed volume is equal to zero,
then there is no solution for these polynomials, and hence there is
no solution for the whole polynomial system. If there are finite
number of solutions for these polynomials, then these solutions
cannot satisfy the rest of the polynomials with probability one
since these polynomials are independent with the rest of the
polynomials. Therefore, there is no solution for the whole
polynomial system almost surely.

\section{New Closed Form Solutions}\label{sec:NewClosedForm}

The closed form solutions of interference alignment for MIMO
interference networks with constant channel coefficients are known
for the cases including 3-user interference network with $M=N$
antennas at each node \cite{85} and the symmetric ${(4\times
8,3)(4\times 8,2)^3}$ system\cite{89}. These closed form solutions
share a common structure: If two interference vectors are aligned at
two different receivers, then there exists an eigenvector solution.
Motivated by this structure, we provide the solutions for the
systems $(2\times3, 1)^2 (3 \times 2, 1)^2$ and $(2\times3, 1)^4$ in
this section. Note that we proved these systems are feasible by
computing the mixed volumes in the previous section. Since for both
systems, each transmitter sends only one beam, we hereafter drop the
subscript ``1" for convenience, i.e., $\xv^{[i]} $ and $\xu^{[i]}$
denote the beamforming vectors of the $i^{th}$ transmitter and
receiver, respectively.

\subsection{Asymmetric $(2\times3, 1)^2 (3 \times 2,1)^2$ System}
We first consider receiver 3 and 4. At receiver 3, the interference
is nulled by the receive filter $\xu^{[3]}$, i.e.,
\begin{eqnarray}\label{eqn:intcondrx3}
\xu^{[3]}\xH^{[31]}\xv^{[1]}=\xu^{[3]}\xH^{[32]}\xv^{[2]}=\xu^{[3]}\xH^{[34]}\xv^{[4]}=0.
\end{eqnarray}
To satisfy the above condition, we align the interference from
transmitter 1 and 2 along the same dimension at receiver 3, i.e.,
\begin{eqnarray}\label{eqn:v1v2atrx3}
\hspace{-.5cm}\text{span}\left(\xH^{[31]}\xv^{[1]}\right)&=&\text{span}\left(\xH^{[32]}\xv^{[2]}\right)\\
\hspace{-.5cm}\Rightarrow
\text{span}\left(\left(\xH^{[32]}\right)^{-1}\xH^{[31]}\xv^{[1]}\right)&=&\text{span}\left(\xv^{[2]}\right)\nonumber,
\end{eqnarray}
where span($\cdot$) denotes the space spanned by the columns of a
matrix.
\par Similarly, at receiver 4, the interference is nulled by
the receive filter $\xu^{[4]}$, i.e.,
\begin{eqnarray}\label{eqn:intcondrx4}
\xu^{[4]}\xH^{[41]}\xv^{[1]}=\xu^{[4]}\xH^{[42]}\xv^{[2]}=\xu^{[4]}\xH^{[43]}\xv^{[3]}=0.
\end{eqnarray}
To satisfy the above condition, we align the interference from
transmitter 1 and 2 along the same dimension at receiver 4, i.e.,
\begin{eqnarray}
\hspace{-.5cm}\text{span}\left(\xH^{[41]}\xv^{[1]}\right)&=&\hspace{-.2cm}\text{span}\left(\xH^{[42]}\xv^{[2]}\right)\label{eqn:v1v2atrx4}\\
\hspace{-.5cm}\Rightarrow
\text{span}\left(\left(\xH^{[42]}\right)^{-1}\xH^{[41]}\xv^{[1]}\right)&=&\text{span}\left(\xv^{[2]}\right).\nonumber
\end{eqnarray}
Notice that both $\xv^{[1]}$ and $\xv^{[2]}$ have to satisfy
\eqref{eqn:v1v2atrx3} and \eqref{eqn:v1v2atrx4}, from which we
obtain
\begin{eqnarray*}
\text{span}\left(\xv^{[1]}\right)=\text{span}\left(\left(\xH^{[41]}\right)^{-1}\xH^{[42]}\left(\xH^{[32]}\right)^{-1}\xH^{[31]}\xv^{[1]}\right).
\end{eqnarray*}
Therefore, we find that $\xv^{[1]}$ is the eigenvector of
$\left(\xH^{[41]}\right)^{-1}\xH^{[42]}\left(\xH^{[32]}\right)^{-1}\xH^{[31]}$.
Then, we determine $\xv^{[2]}$ from \eqref{eqn:v1v2atrx4}. Since
$\xv^{[1]}$ and $\xv^{[2]}$ are determined, we can determine the
receive filters $\xu^{[3]}$ and $\xu^{[4]}$ from
\eqref{eqn:intcondrx3} and \eqref{eqn:intcondrx4}. The interference
alignment we presented so far is shown in Fig.
\ref{fig:IntAlRec3and4}.
\begin{figure} [!t] \centering
{
\begin{pspicture}(0,0)(7,6.4) 

\psframe(1,4.8)(1.6,5.8) \pscircle(1.3,5){.1} \pscircle(1.3,5.6){.1}
\pnode(0,5.3){Tx1} \pnode(1,5){Tx1a}\pnode(1,5.6){Tx1b}
\pnode(1.6,5.3){Tx1c} \rput[l](.1,6){\small{$\xv^{[1]}$}}
\psline[linecolor=red]{->}(.8,5)(.4,5.6)

\psframe(5.4,4.8)(6,5.8) \pscircle(5.7,5){.1} \pscircle(5.7,5.3){.1}
\pscircle(5.7,5.6){.1} \pnode(7,5.3){Rx1}
\pnode(6,5){Rx1a}\pnode(6,5.3){Rx1b}\pnode(6,5.6){Rx1c}
\pnode(5.4,5.3){Rx1d}

\psframe(1,3.2)(1.6,4.2) \pscircle(1.3,3.4){.1} \pscircle(1.3,4){.1}
\pnode(0,3.7){Tx2} \pnode(1,3.4){Tx2a}\pnode(1,4){Tx2b}
\pnode(1.6,3.7){Tx2c} \rput[l](.1,4.0){\small{$\xv^{[2]}$}}
\psline[linecolor=red]{->}(.2,3.4)(.8,4)

\psframe(5.4,3.2)(6,4.2) \pscircle(5.7,3.4){.1}
\pscircle(5.7,3.7){.1} \pscircle(5.7,4){.1} \pnode(7,3.7){Rx2}
\pnode(6,3.4){Rx2a}\pnode(6,3.7){Rx2b}\pnode(6,4){Rx2c}
\pnode(5.4,3.7){Rx2d}

\psframe(1,1.6)(1.6,2.6) \pscircle(1.3,1.8){.1}
\pscircle(1.3,2.1){.1} \pscircle(1.3,2.4){.1} \pnode(0,2.1){Tx3}
\pnode(1,1.8){Tx3a}\pnode(1,2.4){Tx3b} \pnode(1.6,2.1){Tx3c}

\psframe(5.4,1.6)(6,2.6) \pscircle(5.7,1.8){.1}
\pscircle(5.7,2.4){.1} \pnode(7,2.1){Rx3}
\pnode(6,1.8){Rx3a}\pnode(6,2.1){Rx3b}\pnode(6,2.4){Rx3c}
\pnode(5.4,2.1){Rx3d} \psline[linecolor=red]{->}(6.2,1.8)(6.5,2.4)
\psline[linecolor=red]{->}(6.2,1.8)(6.6,2.6)
\psline[linestyle=dashed]{->}(6.2,1.8)(6.8,1.6)
\rput[l](6.8,2.1){\small{$\xu^{[3]}$}}

\psframe(1,0)(1.6,1) \pscircle(1.3,.2){.1}\pscircle(1.3,.5){.1}
\pscircle(1.3,.8){.1} \pnode(0,.5){Tx4}
\pnode(1,.2){Tx4a}\pnode(1,.8){Tx4b} \pnode(1.6,.5){Tx4c}

\psframe(5.4,0)(6,1) \pscircle(5.7,.2){.1} \pscircle(5.7,.8){.1}
\pnode(7,.5){Rx4}
\pnode(6,.2){Rx4a}\pnode(6,.5){Rx4b}\pnode(6,.8){Rx4c}
\pnode(5.4,.5){Rx4d} \psline[linecolor=red]{->}(6.7,.8)(6.75,.3)
\psline[linecolor=red]{->}(6.7,.8)(6.8,0)
\psline[linestyle=dashed]{->}(6.7,.8)(6.2,.7)
\rput[l](6.1,1.0){\small{$\xu^{[4]}$}}

\ncline{->}{Tx1c}{Rx1d}\ncline{->}{Tx1c}{Rx2d}\ncline[linecolor=red]{->}{Tx1c}{Rx3d}\ncline[linecolor=red]{->}{Tx1c}{Rx4d}
\ncline{->}{Tx2c}{Rx2d}\ncline{->}{Tx2c}{Rx1d}\ncline[linecolor=red]{->}{Tx2c}{Rx3d}\ncline[linecolor=red]{->}{Tx2c}{Rx4d}
\ncline{->}{Tx3c}{Rx3d}\ncline{->}{Tx3c}{Rx1d}\ncline{->}{Tx3c}{Rx2d}\ncline{->}{Tx3c}{Rx4d}
\ncline{->}{Tx4c}{Rx4d}\ncline{->}{Tx4c}{Rx1d}\ncline{->}{Tx4c}{Rx2d}\ncline{->}{Tx4c}{Rx3d}
\end{pspicture}
 \caption{Interference alignment between transmitter 1 and 2, and receiver 3 and 4.} \label{fig:IntAlRec3and4}}
\end{figure}
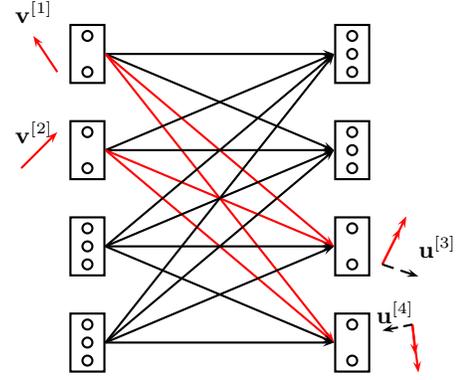
\par Now, we consider the interference from transmitter
4 at receiver 3. Since $\xu^{[3]}$ is already fixed (designed), we
need to zero-force the interference from transmitter 4 at receiver
3, i.e.,
\begin{eqnarray}\label{eqn:ZFRx3}
\xu^{[3]}\xH^{[34]}\xv^{[4]}=0,
\end{eqnarray}
which implies that $\xv^{[4]}$ lies in the null space of
$\xu^{[3]}\xH^{[34]}$. Since $\xu^{[3]}\xH^{[34]}$ is a $1 \times 3$
row vector, it has a null space with dimension 2. Let two $3 \times
1$ vectors $\mathbf{n}_{41}$ and $\mathbf{n}_{42}$ form the basis of
that null space. Then $\xv^{[4]}$ can be expressed as
\begin{eqnarray}\label{eqn:v4}
\xv^{[4]}=v_4'(1)\mathbf{n}_{41}+v_4'(2)\mathbf{n}_{42}=\underbrace{[\mathbf{n}_{41}~
\mathbf{n}_{42}]}_{\mathbf{N}_4}\underbrace{\left[\begin{array}{c}v'_4(1)\\v'_4(2)\end{array}\right]}_{\xv'^{[4]}},
\end{eqnarray}
where $\mathbf{N}_4$ is a $3 \times 2$ matrix and $\xv'^{[4]}$ is a
$2\times 1$ vector. Plugging \eqref{eqn:v4} into \eqref{eqn:ZFRx3},
we obtain
\begin{eqnarray*}
\xu^{[3]}\xH^{[34]}\mathbf{N}_4 \xv'^{[4]}=0.
\end{eqnarray*}
Instead of designing a $3\times 1$ vector $\xv^{[4]}$, now we need
to design a $2\times 1$ vector $\xv'^{[4]}$. Equivalently, we can
also think that transmitter 4 loses 1 antenna, which is illustrated
in Fig. \ref{fig:Tx4}. We leave $\xv'^{[4]}$ to be determined later.
\begin{figure}[!t]
\centering \subfigure[Transmitter 4.] 
{ \label{fig:Tx4}
\begin{pspicture}(0,0)(2.7,1.5) 
\psframe(.8,0)(1.4,1) \cnode(1.1,.2){.1}{A} \cnode(1.1,.8){.1}{B}
\psframe(2.4,0)(3,1) \cnode(2.7,.2){.1}{C} \cnode(2.7,.5){.1}{D}
\cnode(2.7,.8){.1}{E} \ncline{A}{C}\ncline{A}{D}\ncline{A}{E}
\ncline{B}{C}\ncline{B}{D}\ncline{B}{E} \pnode(0,.5){Tx1a}
\ncline[linestyle=dashed]{Tx1a}{A}\ncline[linestyle=dashed]{Tx1a}{B}
\psellipse[linestyle=dashed](.5,.5)(.2,.5)
\psellipse[linestyle=dashed](1.9,.5)(.2,.5)
\rput[c](.5,1.2){\small{$\xv^{\prime[4]}$}}
\rput[c](1.9,1.2){\small{$\mathbf{N}_4$}}
\end{pspicture}
} \centering \subfigure[Receiver 1.] 
{ \label{fig:Rx1}
\begin{pspicture}(0,0)(3.8,1.2) 
\psframe(.8,0)(1.4,1) \cnode(1.1,.2){.1}{A} \cnode(1.1,.5){.1}{B}
\cnode(1.1,.8){.1}{C} \psframe(2.4,0)(3,1) \cnode(2.7,.2){.1}{D}
\cnode(2.7,.8){.1}{E} \ncline{A}{D}\ncline{A}{E}\ncline{B}{D}
\ncline{B}{E}\ncline{C}{D}\ncline{C}{E} \pnode(3.8,.5){Tx1a}
\ncline[linestyle=dashed]{Tx1a}{D}\ncline[linestyle=dashed]{Tx1a}{E}
\psellipse[linestyle=dashed](3.3,.5)(.2,.5)
\psellipse[linestyle=dashed](1.9,.5)(.2,.5)
\rput[c](3.3,1.2){\small{$\xu^{\prime[1]}$}}
\rput[c](1.9,1.2){\small{$\mathbf{P}_1$}}
\end{pspicture}
}
 \caption{Equivalent representations of transmitter 4 $\left(\xv^{[4]}=\mathbf{N}_4\xv'^{[1]}\right)$ and receiver 1 $\left(\xu^{[1]}=\xu'^{[1]}\mathbf{P}_1\right)$.}
 \label{fig:EqReps}
\end{figure}
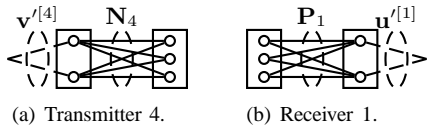
\par Now, we consider the interference from transmitter
3 at receiver 4. Since $\xu^{[4]}$ is already fixed (designed), we
need to zero-force the interference from transmitter 3 at receiver
4, i.e.,
\begin{eqnarray}\label{eqn:ZFRx4}
\xu^{[4]}\xH^{[43]}\xv^{[3]}=0,
\end{eqnarray}
which implies that $\xv^{[3]}$ lies in the null space of
$\xu^{[4]}\xH^{[43]}$. Therefore, following the same approach as
$\xv^{[4]}$, $\xv^{[3]}$ can be expressed as
\begin{eqnarray*}
\xv^{[3]}= \mathbf{N}_3 \xv'^{[3]},
\end{eqnarray*}
where $\mathbf{N}_3$ is a $3 \times 2$ matrix whose columns form the
basis of the null space of $\xu^{[4]}\xH^{[43]}$ and $\xv'^{[3]}$ is
a $2\times 1$ vector. Again, instead of designing a $3\times 1$
vector $\xv^{[3]}$, now we need to design a $2\times 1$ vector
$\xv'^{[3]}$. Similarly, we can also think that transmitter 3 loses
1 antenna.
\par So far, we have considered all interference at receiver
3 and 4. As a result, we determined
$\xv^{[1]},\xv^{[2]},\xu^{[3]},\textrm{ and } \xu^{[4]}$. We left
$\xv'^{[3]}$ and $\xv'^{[4]}$ to be determined later. Now, let us
consider receiver 1 and 2.
\par At receiver 1 and 2, the interference is nulled by the receive
filters $\xu^{[1]}$ and $\xu^{[2]}$, i.e.,
\begin{equation}
\xu^{[1]}\xH^{[12]}\xv^{[2]}=\xu^{[1]}\xH^{[13]}\xv^{[3]}=\xu^{[1]}\xH^{[14]}\xv^{[4]}=0
\label{eqn:intcondrx1}
\end{equation}
and
\begin{equation}
\xu^{[2]}\xH^{[21]}\xv^{[1]}=\xu^{[2]}\xH^{[23]}\xv^{[3]}=\xu^{[2]}\xH^{[24]}\xv^{[4]}=0\label{eqn:intcondrx2},
\end{equation}
respectively. As we did previously, one may directly want to
determine $\xv^{[3]}$, $\xv^{[4]}$, $\xu^{[1]}$ and $\xu^{[2]}$ from
the above equations by writing similar alignment conditions in
\eqref{eqn:v1v2atrx3} and \eqref{eqn:v1v2atrx4}. However, notice
that $\xv^{[1]}$ and $\xv^{[2]}$ are already fixed (designed).
Therefore, we need to null the interference from transmitter 2 and
1, i.e.,
\begin{equation*}
\xu^{[1]}\xH^{[12]}\xv^{[2]}=0
\end{equation*}
and
\begin{equation*}
\xu^{[2]}\xH^{[21]}\xv^{[1]}=0,
\end{equation*}
respectively. The above equations imply that $\xu^{[1]}$ and
$\xu^{[2]}$ lay in the left null space of $\xH^{[12]}\xv^{[2]}$ and
$\xH^{[21]}\xv^{[1]}$, respectively. Let $\mathbf{P}_1$ be a $2
\times 3$ matrix whose two rows are orthogonal to
$\xH^{[12]}\xv^{[2]}$. Then, $\xu^{[1]}$ can be expressed as
\begin{eqnarray*}
\xu^{[1]}= \xu'^{[1]} \mathbf{P}_1,
\end{eqnarray*}
where $\xu'^{[1]}$ is a $1 \times 2$ vector. As a result, instead of
designing a $1 \times 3$ vector $\xu^{[1]}$, we need to design a
$1\times 2$ vector $\xu'^{[1]}$. Equivalently, we can also think
that receiver 1 loses 1 antenna, which is illustrated in Fig.
\ref{fig:Rx1}. Similarly, $\xu^{[2]}$ can be expressed as
\begin{eqnarray*}
\xu^{[2]}= \xu'^{[2]} \mathbf{P}_2,
\end{eqnarray*}
where $\xu'^{[2]}$ is a $1\times 2$ vector and $\mathbf{P}_2$ is a
$2 \times 3$ matrix whose two rows are orthogonal to
$\xH^{[21]}\xv^{[1]}$.
\par Now, the interference alignment
conditions (\ref{eqn:intcondrx1}) and (\ref{eqn:intcondrx2}) are
equivalent to
\begin{equation}\label{eqn:v3v4atrx1}
\xu'^{[1]}\underbrace{\mathbf{P}_1\xH^{[13]}\mathbf{N}_3}_{\xH'^{[13]}}
\xv'^{[3]}=\xu'^{[1]}\underbrace{\mathbf{P}_1\xH^{[14]}\mathbf{N}_4}_{\xH'^{[14]}}
\xv'^{[4]}=0,
\end{equation}
and
\begin{equation}\label{eqn:v3v4atrx2}
\xu'^{[2]}\underbrace{\mathbf{P}_2\xH^{[23]}\mathbf{N}_3}_{\xH'^{[23]}}
\xv'^{[3]}=\xu'^{[2]}\underbrace{\mathbf{P}_2\xH^{[24]}\mathbf{N}_4}_{\xH'^{[24]}}
\xv'^{[4]}=0,
\end{equation}
respectively where $\xH'^{[13]}$, $\xH'^{[14]}$, $\xH'^{[23]}$, and
$\xH'^{[24]}$ are $2 \times 2$ matrices.
\par Similar to $\xv^{[1]}$, we find that $\xv'^{[3]}$ is the
eigenvector of
$\left(\xH'^{[23]}\right)^{-1}\xH'^{[24]}\left(\xH'^{[14]}\right)^{-1}\xH'^{[13]}$.
Then, we determine $\xv'^{[4]}$ and $\xu'^{[1]}$ from
(\ref{eqn:v3v4atrx1}) and $\xu'^{[2]}$ from (\ref{eqn:v3v4atrx2}).
The interference alignment we presented for receiver 1 and 2 is
shown in Fig. \ref{fig:IntAlRec1and2}.
\par As a result, we completed designing all transmit and receive beamforming filters in the $(2\times3, 1)^2 (3 \times 2, 1)^2$ system.
\begin{figure} [!t] \centering {
\begin{pspicture}(0,0)(7.25,6.4) 

\psframe(2.2,4.8)(2.8,5.8) \pscircle(2.5,5){.1}
\pscircle(2.5,5.6){.1} \pnode(2.8,5.3){Tx1c}

\psframe(4.2,4.8)(4.8,5.8) \cnode(4.5,5){.1}{1A}
\cnode(4.5,5.3){.1}{1B} \cnode(4.5,5.6){.1}{1C} \pnode(5.8,5.3){Rx1}
\pnode(4.2,5.3){Rx1d} \rput[c](7.1,5.6){\small{$\xu'^{[1]}$}}
\psframe(5.4,4.8)(6,5.8) \cnode(5.7,5){.1}{1D}
 \cnode(5.7,5.6){.1}{1E} \ncline{1A}{1D}\ncline{1A}{1E}
 \ncline{1B}{1D}\ncline{1B}{1E}\ncline{1C}{1D}\ncline{1C}{1E}
 \psellipse[linestyle=dashed](5.1,5.3)(.2,.5)

\psframe(2.2,3.2)(2.8,4.2) \pscircle(2.5,3.4){.1}
\pscircle(2.5,4){.1} \pnode(2.8,3.7){Tx2c}

\psframe(4.2,3.2)(4.8,4.2) \cnode(4.5,3.4){.1}{2A}
\cnode(4.5,3.7){.1}{2B} \cnode(4.5,4){.1}{2C} \pnode(4.2,3.7){Rx2d}
\rput[c](7.1,3.6){\small{$\xu'^{[2]}$}} \psframe(5.4,3.2)(6,4.2)
\cnode(5.7,3.4){.1}{2D} \cnode(5.7,4){.1}{2E}
\ncline{2A}{2D}\ncline{2A}{2E}
 \ncline{2B}{2D}\ncline{2B}{2E}\ncline{2C}{2D}\ncline{2C}{2E}
 \psellipse[linestyle=dashed](5.1,3.7)(.2,.5)

\psframe(2.2,1.6)(2.8,2.6) \cnode(2.5,1.8){.1}{3A}
\cnode(2.5,2.1){.1}{3B} \cnode(2.5,2.4){.1}{3C}
\pnode(2.8,2.1){Tx3c} \rput[l](.3,2.8){\small{$\xv'^{[3]}$}}
\psline[linecolor=blue]{->}(.8,.4)(.4,.8) \psframe(1,1.6)(1.6,2.6)
\cnode(1.3,1.8){.1}{3D} \cnode(1.3,2.4){.1}{3E}
\ncline{3A}{3D}\ncline{3A}{3E}
 \ncline{3B}{3D}\ncline{3B}{3E}\ncline{3C}{3D}\ncline{3C}{3E}
 \psellipse[linestyle=dashed](1.9,2.1)(.2,.5)

\psframe(4.2,1.6)(4.8,2.6) \pscircle(4.5,1.8){.1}
\pscircle(4.5,2.4){.1} \pnode(4.2,2.1){Rx3d}
\psline[linecolor=blue]{->}(6.4,5)(6.25,5.6)
\psline[linecolor=blue]{->}(6.4,5)(6.2,5.8)
\psline[linestyle=dashed]{->}(6.4,5)(7,5.2)

\psframe(2.2,0)(2.8,1) \cnode(2.5,.2){.1}{4A}\cnode(2.5,.5){.1}{4B}
\cnode(2.5,.8){.1}{4C} \pnode(2.8,.5){Tx4c}
\rput[l](.3,1.2){\small{$\xv'^{[4]}$}}
\psline[linecolor=blue]{->}(.2,2)(.8,2.4) \psframe(1,0)(1.6,1)
\cnode(1.3,.2){.1}{4D} \cnode(1.3,.8){.1}{4E}
\ncline{4A}{4D}\ncline{4A}{4E}
\ncline{4B}{4D}\ncline{4B}{4E}\ncline{4C}{4D}\ncline{4C}{4E}
 \psellipse[linestyle=dashed](1.9,.5)(.2,.5)

\psframe(4.2,0)(4.8,1) \pscircle(4.5,.2){.1} \pscircle(4.5,.8){.1}
\pnode(4.2,.5){Rx4d}

\psline[linecolor=blue]{->}(6.2,3.4)(6.5,4.2)
\psline[linecolor=blue]{->}(6.2,3.4)(6.43,4)
\psline[linestyle=dashed]{->}(6.2,3.4)(6.8,3.2)

\ncline{->}{Tx1c}{Rx1d}\ncline{->}{Tx1c}{Rx2d}\ncline{->}{Tx1c}{Rx3d}\ncline{->}{Tx1c}{Rx4d}
\ncline{->}{Tx2c}{Rx2d}\ncline{->}{Tx2c}{Rx1d}\ncline{->}{Tx2c}{Rx3d}\ncline{->}{Tx2c}{Rx4d}
\ncline{->}{Tx3c}{Rx3d}\ncline[linecolor=blue]{->}{Tx3c}{Rx1d}\ncline[linecolor=blue]{->}{Tx3c}{Rx2d}\ncline{->}{Tx3c}{Rx4d}
\ncline{->}{Tx4c}{Rx4d}\ncline[linecolor=blue]{->}{Tx4c}{Rx1d}\ncline[linecolor=blue]{->}{Tx4c}{Rx2d}\ncline{->}{Tx4c}{Rx3d}
\end{pspicture}
 \caption{Interference alignment between transmitter 3 and 4, and receiver 1 and 2.}\label{fig:IntAlRec1and2}}
\end{figure}
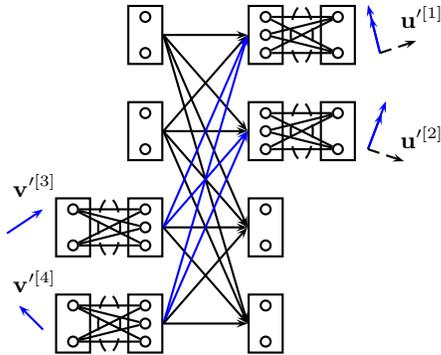

\subsection{Symmetric $(2\times 3, 1)^4$ System}

It is difficult to directly express the closed form solution for the
$(2\times 3, 1)^4$ system. However, by first presenting a closed
form solution for the $(2\times 4, 1)(2\times 3, 1)^3$ system, where
there is an extra freedom (i.e., an extra receive antenna), we show
the solution for the $(2\times 3, 1)^4$ system.

For the $(2\times 4, 1)(2\times 3, 1)^3$ system, suppose that we
randomly pick the beamforming vector $\mathbf{v}^{[1]}$ at
transmitter 1. To eliminate the interference caused by transmitter 1
at receiver 2, 3 and 4, each receiver needs to discard the dimension
occupied by this interference. In other words, each receiver only
uses the 2 dimensional subspace, which is orthogonal to the
direction of the interference caused by transmitter 1. Equivalently,
we can think that each of the receivers 2, 3, and 4 loses 1 antenna
as illustrated in the previous subsection. Now, the transmitter and
receiver pairs 2, 3, and 4 are equivalent to a 3-user interference
channel with 2 antennas at each node for which the closed form
solution is known \cite{85}. Finally, since the receiver 1 has 4
antennas, it can separate the desired signal by zero-forcing all the
interference. As a result, each user achieves 1 DoF. We omit the
explicit closed form solution for this system due to brevity concern
of the paper.

Next, we intuitively show that the $(2\times 3, 1)^4$ system has
also a closed form solution by using the $(2\times 4, 1)(2\times 3,
1)^3$ system for which we previously showed that there exists a
closed form solution. Now, for the $(2\times 4, 1)(2\times 3, 1)^3$
system, consider the receive beamforming vector at receiver 1:
$$\mathbf{u}^{[1]}=\left[
\begin{array}{c}
1 \\
u_{1}\\
u_{2}\\
u_{3}\\
\end{array}%
\right].$$ An extra variable at receiver 1 due to an extra antenna
gives us a freedom to chose it arbitrarily. We pick $u_3$ as this
extra variable. However, instead of choosing $u_3$ arbitrarily, we
choose $u_3$ in terms of other variables by iteratively solving it
from the equations that it is involved\footnote{The iterative
solution for the variable $u_3$ can be clearly seen in the
corresponding polynomial system of the $(2\times 4, 1)(2\times 3,
1)^3$ system.}: \par $\mathbf{u}^{[1]}$ is involved in 3 equations,
which are
$$E_{11}^{12}, E_{11}^{13}, \textrm{ and } E_{11}^{14}$$
(the equations from transmitter 2, 3, and 4 to receiver 1). We
iteratively solve $u_3$ in terms of other variables, which are
$$v_2, v_3, \textrm{ and }v_4$$ \big($\mathbf{v}^{[2]}=\left[
\begin{array}{c}
1 \\
v_{2}%
\end{array}%
\right] $, $\mathbf{v}^{[3]}=\left[
\begin{array}{c}
1 \\
v_{3}%
\end{array}%
\right] $, and $\mathbf{v}^{[4]}=\left[
\begin{array}{c}
1 \\
v_{4}%
\end{array}%
\right] $\big).\\

As a result, eliminating the extra variable $u_3$ by solving it
iteratively in terms of other variables provides us the solution for
the $(2\times 3, 1)^4$ system, although it is difficult to express
it in a closed form.

\section{Multi-Beam Cases}\label{sec:GOuterB}

The solvability of polynomial systems for the multi-beam cases is
more involved as explained in Section \ref{sec:BTheorem}. Only the
proper system definition itself cannot state the feasibility of
these cases since this definition does not consider the dependency
of coefficients. Note that even the current advancements in
algebraic geometry are insufficient for these cases. At this point,
we can use information theoretic outer bounds (general and
cooperative outer bounds), which we explain next, in addition to our
proper system condition to test the feasibility of these cases.

It is well known that for a point to point MIMO channel with $M$
transmit and $N$ receive antennas, DoF is $\min(M,N)$ \cite{36,93}.
In addition, for a 2-user MIMO interference channel with
$M^{[1]},M^{[2]}$ and $N^{[1]},N^{[2]}$ antennas at transmitters and
receivers, respectively, it is well known that DoF is
$\min\big(M^{[1]}+M^{[2]}, N^{[1]}+N^{[2]},
\max(M^{[1]},N^{[2]}),\max(M^{[2]},N^{[1]})\big)$ \cite{93}. These
two results serve as general DoF outer bounds for a $K\text{-user}$
MIMO interference network:
\begin{align}
d^{[i]}&\leq\min(M^{[i]},N^{[i]})\label{eqn:outerbound1}\\
d^{[i]}+d^{[j]}&\leq \min\big(M^{[i]}+M^{[j]}, N^{[i]}+N^{[j]},\notag\\
&\max(M^{[i]},N^{[j]}),\max(M^{[j]},N^{[i]})\big)\label{eqn:outerbound2}\\
\textrm{ for all } i,j\in\mathcal{K}&. \nonumber
\end{align}

As mentioned before, we assume that the first condition
\eqref{eqn:outerbound1} is always satisfied even if it is not
explicitly stated every time throughout the paper.

\example Although the $(3\times3, 2)^2$ system is proper, this
system is almost surely infeasible since it does not satisfy the
general outer bound \eqref{eqn:outerbound2}.

Another outer bound, which is trivially obtained by using
\eqref{eqn:outerbound2} is the cooperative outer bound. That is, the
general outer bound \eqref{eqn:outerbound2} can also be used for all
combinations of cooperation \emph{within} transmitters and
\emph{within} receivers in a $K\text{-user}$ MIMO interference
network. If the general outer bound \eqref{eqn:outerbound2} is not
satisfied for any of these combinations, then the system is almost
surely infeasible.

\example Consider a 4-user MIMO interference network with
$M^{[1]},\cdots,M^{[4]}$ antennas and $N^{[1]},\cdots,N^{[4]}$
antennas at transmitters and receivers, respectively. The general
outer bound \eqref{eqn:outerbound2} can also be used for the 3-user
cooperative case of this network with
$M^{[1]}+M^{[2]},M^{[3]},M^{[4]}$ and ${N^{[1]}+N^{[2]}}$,
$N^{[3]},N^{[4]}$ antennas at transmitters and receivers,
respectively. In addition, it can also be used for the 2-user
cooperative case of this network with
$M^{[1]}+M^{[2]},M^{[3]}+M^{[4]}$ and ${N^{[1]}+N^{[2]}}$,
$N^{[3]}+N^{[4]}$ antennas at transmitters and receivers,
respectively. These are the only 2 cooperative cases of the original
4-user network and the general outer bound \eqref{eqn:outerbound2}
can be checked for all cooperative cases.

\example Consider the $(3\times4,2)(1\times3,1)(10\times4,2)$
system, which is proper and which satisfies the general outer bound
\eqref{eqn:outerbound2}. Now, consider the cooperative case between
the first and second users; that is, consider the
$(4\times7,3)(10\times4,2)$ system. Since the general outer bound
\eqref{eqn:outerbound2} for this cooperative case (briefly,
cooperative outer bound) is not satisfied, this system is almost
surely infeasible.

Next, we list some examples for ${K\textrm{-user}}$ interference
networks with more than $K$ DoF, all of which satisfy the general
outer bound \eqref{eqn:outerbound2} and the cooperative outer bound.
We discuss the feasibility or infeasibility of these systems
depending on the proper system condition. These examples highlight
the usefulness of inequalities from Theorems \ref{theo:symmetric}
and \ref{theo:asymmetric}, and the usefulness of grouping.

\example \label{ex:5x5,2}Consider the $\left(5\times 5, 2\right)^4$
system. There are ${N_e=48}$ equations in total; therefore, there
are $2^{48}-1$ subsets of equations. Testing each of them could be
very challenging due to the proper system definition
\eqref{eqn:proper}. However, the system is easily seen to be proper
from Theorem \ref{theo:symmetric} since ${M+N-(K+1)d=5+5-10=0}$.
Note that $(2\times 8,2)^4,(3\times 7,2)^4,(4\times 6,2)^4,(5\times
5,2)^4,(6\times 4,2)^4,$ $(7\times 3,2)^4, \textrm{ and } (8\times
2, 2)^4$ all belong to the same group, where any system in the group
can be obtained by successively transferring an antenna between
transmitters and receivers.

\example \label{ex:improperasymmetric} Consider the $(5\times
5,3)(5\times 5,2)^3$ system. There are $N_e=60$ equations in total;
therefore, there are ${2^{60}-1}$ subsets of equations. Testing each
of them could be very challenging due to the proper system
definition \eqref{eqn:proper}. However, the system is easily seen to
be improper from Theorem \ref{theo:asymmetric} since the total
number of variables $N_v=48$ is less than the total number of
equations $N_e=60$.

Finally, note that there are two important features in a polynomial
system that can lead the polynomial system to solvability or
non-solvability: The coefficients and the structure of polynomial
system. That is, one can lead the system to solvability or
non-solvability by deliberately selecting the coefficients or by
deliberately introducing a structure to the polynomial system.
Bernshtein's theorem captures the structure of a multivariate
polynomial system by finding the mixed volume of the Newton
polytopes of the multivariate polynomial system. Therefore,
Bernshtein's theorem only requires the independency of coefficients.
Otherwise, if the coefficients are dependent, Bernshtein's theorem
provides only an upper bound for the number of solutions of a
multivariate polynomial system. Following a similar argument, due to
the nature of our proper system definition, the proper system
condition cannot handle the cases when the coefficients are
dependent or when the polynomial system has a certain structure. For
the latter case, consider the example mentioned in Remark
\ref{rem:Rem2}. For \emph{diagonal} (time-varying) channels, the DoF
of a $K\text{-user}$ MIMO network ($M=N$ antennas at each node) is
$K/2$ times the number of DoF achieved by each user in the absence
of interference \cite{85}. Note that for the corresponding
polynomial system, $N_e>N_v$. Although interference networks with
diagonal channels are improper (because $N_e>N_v$), the interference
alignment is feasible since the diagonal channels obviously bring a
structure to the polynomial system, which leads the system to
solvability.

\section{Conclusion}\label{sec:Conclusion}

In this paper, we explore the feasibility of interference alignment
through beamforming in MIMO interference networks. Accordingly, we
consider the alignment problem for an interference network as the
solvability of its corresponding multivariate polynomial system.
Ideally, we would like to find the conditions that would show the
direct link between the feasibility (infeasibility) of an
interference network and the solvability (non-solvability) of its
corresponding polynomial system. For single beam cases, our results
indicate that the solvability of corresponding polynomial systems is
based on counting the number of equations and variables in the
polynomial systems. We support our intuition by providing numerical
results for a variety of cases, by presenting closed form solutions
for new systems, and by providing rigorous proofs for some important
cases.

On the other hand, for multi-beam cases, the current advancements in
algebraic geometry are insufficient to prove the solvability of
corresponding polynomial systems. Based on numerical results, we
show that the connection between feasible and proper systems can be
further strengthened by including information theoretic (general and
cooperative) outer bounds to our proper system condition. In
addition, based on numerical results, we also observe that if the
system is improper, then it is infeasible.

\section*{Appendix\\Genericity}\label{sec:Appendix}

The term genericity in algebraic geometry has a mathematical
explanation, which is beyond our scope (see ``generic property" from
Wikipedia). In Chapter 7 of \cite{94}, the proof that the polynomial
system ${f_1=0,\cdots,f_n=0}$ has \emph{generically}
$\textrm{MV}(P_1,\cdots,P_n)$ number of common solutions is shown by
induction on $n$, which denotes the dimension of a polynomial
system. Here, we will show that genericity implies independent
random coefficients, which matters for our scope in this paper. For
this purpose, we first start with the definition of coefficient
polynomial. Note that mathematics literature does not directly and
simply present the genericity in the matter of our scope as we
present in this Appendix.
\par Let $p(c)$ denote the coefficient polynomial, which is
dependent on the coefficients of polynomials $f_1,\cdots,f_n$, where
\begin{equation*}
c\subset C=\{c_{ij}|~ i \in \{1,\cdots,n\},j \in \{1,\cdots,m_i\}\},
\end{equation*}
is the subset of all coefficients of polynomial system. Next, we
define the term algebraic independence of coefficients, which is
originally related to the term algebraic independence in algebraic
geometry.

\definition Let $c$ denote the subset of all coefficients of polynomials
$f_1,\cdots,f_n$. The subset is called algebraically dependent if
\emph{there is a} coefficient polynomial satisfying the equality
$p(c)=0$. Otherwise, it is called algebraically independent.

\example Consider the polynomial ${f=c_1x^2+c_2x+c_3}$.

\begin{itemize}
\item $c=\{c_1\}$ is algebraically dependent if there is a coefficient polynomial satisfying $p(c_1)=0$, e.g.,
${p(c_1)=c_1^2+2c_1}=0$. That is, $c_1$ is not transcendental.
\item $c=\{c_1,c_2\}$ is algebraically dependent if there is a coefficient polynomial satisfying $p(c_1,c_2)=0$,
e.g., $p(c_1,c_2)=c_1^2+c_2=0$.
\end{itemize}

If a polynomial is equal to zero, it is also called a ``vanishing
polynomial" in mathematics, e.g., $p(c)=0$. \par Next, we define
\emph{genericity}, which we rephrase from the definitions 5.6 and
5.3 in Chapter 3 and 7 of \cite{94}, respectively. Note once again
that Bezout's and Bernshtein's theorems give the exact number of
common solutions when the coefficients are generic.

\definition \label{def:generically} A property is said to hold generically for the polynomials $f_1,\cdots,f_n$ if \emph{there is a} coefficient polynomial $p(c)$ such that the nonvanishing of $p(c)$ implies that this property holds.

Intuitively, this definition means that the property for all
polynomials holds for \emph{most} of the coefficients; that is, for
those coefficients satisfying $p(c)\neq0$.

\example\label{ex:generic} Consider $f=c_1x^2+c_2x+c_3=0$, which has
a mixed volume 2. One can claim that the property ``$f$ has two
(distinct) solutions" holds generically. To prove this, we must find
a coefficient polynomial, whose nonvanishing implies this desired
property. The condition is easily seen to be the nonvanishing
discriminant of polynomial,
$p(c)=\textrm{Disc}(f)=c_1(c_2^2-4c_1c_3)\neq0$, which is satisfied
always if the set of all coefficients is algebraically independent.
\emph{Some cases do not require the set of all coefficients to be
algebraically independent}. For example, for the same polynomial,
one can also claim that the property ``$f$ has two solutions with
multiplicities counted" (i.e., the solutions may not be distinct
this time) holds generically. The coefficient polynomial
$p(c)=c_1c_3\neq0$ implies this desired property. As a result, we
briefly say that $f$ has \emph{generically} 2 solutions.

As mentioned before, the proof that the polynomial system has
generically $\textrm{MV}(P_1,\cdots,P_n)$ number of common solutions
is shown by induction on $n$ in \cite{94}. We leave the further
details to be inquired in \cite{94}.

As a result, based on Definition \ref{def:generically}, we can
simply argue that independent random coefficients are almost surely
generic since the set of independent random coefficients is
algebraically independent; that is, $p(c)\neq0$ almost surely, where
$c$ is the set of independent random coefficients.

The proof that the genericity in Bernshtein's theorem implies
independent random coefficients is also shown from an algebraic
point of view in Section 2 of \cite{97}.

\section*{Acknowledgment}
The first author would like to thank V. R. Cadambe at University of
California Irvine for valuable discussions on interference
alignment.
\bibliographystyle{IEEEtran}
\bibliography{IEEEabrv,IEEEfull}
\end{document}